\documentclass[12pt,a4paper]{article}
\usepackage{cite}
\usepackage[T1]{fontenc}
\usepackage{lmodern}
\usepackage{amsmath}
\usepackage{amsfonts}
\usepackage{amssymb}
\usepackage{mathtools}
\usepackage{latexsym}
\usepackage{parskip}
\usepackage{enumerate}
\usepackage{lipsum}
\usepackage[latin1]{inputenc}
\usepackage{amsmath}
\usepackage{amsfonts}
\usepackage{amssymb}
\usepackage{latexsym}

\usepackage{verbatim}

\usepackage{multicol}

\usepackage{parskip}
\usepackage{fullpage}
\usepackage[pdftex]{graphicx}
\usepackage{enumerate}
\numberwithin{equation}{section}

\usepackage{amsthm}

\newtheorem{definicao}{Definition}[section]
\newtheorem{teorema}{Theorem}[section]

\newtheorem{corolario}[teorema]{Corollary}
\newtheorem{lema}[teorema]{Lemma}
\newtheorem{rem}[teorema]{Remark}

\newcounter{exemplo}[section]
\newcommand{\exemplo}{\stepcounter{exemplo}
	\noindent\textbf{Example \arabic{section}.\arabic{exemplo}. }}

\newcommand{\er}{\mathbb{R}}
\newcommand{\en}{\mathbb{N}}

\newcommand{\Ne}{\mathrm{e}}
\newcommand{\Es}{\hspace{0.5cm}}
\newcommand{\ov}{\overline}
\newcommand{\un}{\underline}
\newcommand{\dst}{\displaystyle}
\newcommand{\til}[1]{\widetilde{#1}}

\usepackage[usenames,dvipsnames,svgnames,table]{xcolor}

\DeclareMathOperator{\codim}{codim}
\DeclareMathOperator{\Imagem}{Im}
\DeclareMathOperator{\Ker}{Ker}
\DeclareMathOperator{\Dom}{Dom}
\begin{document}
	
	\begin{center}
		{\bf {\Large Generalized non-autonomous Cohen-Grossberg neural network model}}
	\end{center}

	\begin{center}
		Ahmed Elmwafy, Jos\'{e} J. Oliveira, C\'esar M. Silva
	\end{center}

	\vskip .5cm

	\begin{abstract}
		
	In the present paper, we investigate both the global exponential stability and the existence of a periodic solution of a general differential equation with unbounded distributed delays. The main stability criterion depends on the dominance of the non-delay terms over the delay terms. The criterion for the existence of a periodic solution is obtained with the application of the coincide degree theorem. We use the main results to get criteria for the existence and global exponential stability of periodic solutions of a generalized higher-order periodic Cohen-Grossberg neural network model with discrete-time varying delays and infinite distributed delays.
    Additionally, we provide a comparison with the results in the literature and a numerical simulation to illustrate the effectiveness of some of our results.
	\end{abstract}
	
	
	\noindent
	{\it \textbf{Mathematics Subject Classification System 2020}}: 34K20, 34K25, 34K60, 92B20.
	
	\section{Introduction}
	In the past decades, due to application in various sciences, delayed functional differential equations have attracted the attention of an increasing number of researchers. In many fields, such as population dynamics, ecology, epidemiology, disease evolution, and neural networks, differential equations with delay have served as models.\\
As a result of their widespread use in several fields including image and signal processing \cite{liu2019stability}, pattern recognition \cite{yao1990model}, optimization \cite{park2005lmi}, and content-addressable memory \cite{wu2010exponential}, delayed neural networks have had their dynamical behaviours extensively studied \cite{zhang2020dynamics}, \cite{wu2022time}, \cite{huang2020stability}. \\Obtaining results about the convergence characteristics of neural networks is crucial in these applications. To keep the entire network from acting chaotically, convergent dynamics are required. Significantly, the global convergent dynamics imply that every trajectory of the network can converge to some equilibrium state or invariable sets so that, when used as an associative memory, every state in the underlying space can serve as a key to recover certain stored memory. As an outcome, the state space is entirely covered by different basins of the stored memories. Furthermore, the globally convergent dynamics indicate that the neural network algorithm will ensure convergence to an optimal solution from each initial guess when used as an optimization solver \cite{eliasmith2005unified}. \\
The fact that the connectivity weights, the neuron charging time, and the external inputs change throughout time is another important consideration. Thus, it is relevant to introduce and investigate neural network models that incorporate the temporal structure of neural activities.
%The stability of delayed neural network models has been extensively investigated (for example, see \cite{liang2014existence}, \cite{liu2008new}, \cite{long2007existence}, \cite{oliveira2011global},\cite{oliveira2017convergence}, \cite{wu2010exponential}, \cite{xiang2009almost}). Aside from the publications mentioned above, there is a significant direction of research dealing with the global stability of delayed neural network models. We emphasize, however, that the standard approach in the literature for studying the global asymptotic stability of a system's equilibrium relies on the application of the Lyapunov functional technique. Constructing a Lyapunov functional for a concrete n-dimensional functional differential equation is a challenging task in general. A new Lyapunov functional is required for each model under evolution. Counter to the usual, our technique does not depend on Lyapunov functional. \\

Among the various neural network models that have been extensively investigated and applied, Cohen-Grossberg which was first introduced and investigated by Cohen and Grossberg \cite{cohen1983absolute} by the following system of ordinary differential equations,
\begin{align}
	\frac{dx_i(t)}{dt}&= -a_i(x_i(t))\big[b_i(x_i(t))-\sum_{j=1}^{n} c_{ij}f_j(x_j(t))+I_i(t)\big],\hspace{0.2cm} t\geq 0,\hspace{0.1cm} i=1,\dots,n
\end{align}
where $n$ is a natural number indicates the number of neurons, $x_i(t)$ is the $i$th neuron state at time t, $a_i(u)$ denote the amplification functions, $b_i(u)$ are the self-signal functions, $f_j(u)$ are the activation functions, $c_{ij}$ represent the strengths of connectivity between neurons $i$ and $j$, $I_i$ denote the inputs from outside of the system.\\
Differential equations modelling neural networks should include time delays due to synaptic transmission time across neurons or, in artificial neural networks, communication time among amplifiers in order to be more realistic.\\

Since Cohen and Grossberg first proposed the CGNN model \cite{cohen1983absolute}, the dynamical properties of CGNNs such as stability, instability, and periodic oscillation have been extensively studied for theoretical and application considerations. Some studies have already accomplished several positive results such as \cite{aouiti2020new}, \cite{chen2019multiple}, \cite{li2016some}, \cite{kang2015global}, \cite{dong2021global} and etc., most of the  results in the literature require either the boundedness of the activation functions or the boundedness of delays. For example, \cite{chen2010global} investigated the global exponential stability of the periodic solutions of delayed CGNNs but in the case of discontinuous activation functions. Besides that the existence, uniqueness and stability of almost periodic solutions for a class of NNs have  been studied in \cite{qin2013global}. Meanwhile, \cite{long2007existence}, \cite{zhang2007existence}, and \cite{liu2011periodic} started studying the existence and exponential stability of high-order CGNNs depending on many techniques. For example, \cite{long2007existence} and \cite{zhang2007existence} used some differential inequality techniques, and \cite{liu2011periodic}  depended on using a proper Lyapunov function and the properties of M-matrix. Therefore, the present work is meaningful and the conclusion is novel.

Since as far as we know, there are few results on high-order CGNNs without using the Lyapunov technique, neither assuming the boundedness nor the discontinuity of the activation functions. Motivated by the proceeding studies, we consider a generalized high-order CGNN model with discrete time-varying and distributed delays to study the existence of periodic solutions and global exponential stability without using the Lyapunov technique nor the boundedness of the activation functions.

In this paper, we use the continuation theorem of coincidence degree theory to show the existence of a periodic solution of a generalized system of high-order CGNNs, and then we present sufficient conditions to guarantee the global exponential stability of that system. The remainder of this work is organized as follows. Section 2 is a preliminary section where we introduce our notations and our hypotheses. Section 4 introduces the global exponential stability of general neural network models. In Section 3, we investigate the existence and global exponential stability of the periodic solution of that generalized high-order CGNNs system under certain assumptions. In section 5, We show numerical simulations to demonstrate the efficacy of the results we have obtained.

   \section{Preliminaries and model description}
     In the present paper, for $n\in\en$, we consider the n-dimensional vector space $\er^n$  equipped with the norm $|x|=\max\{|x_i|,\, i=1,\dots,n\}$.

     For a positive real number $\epsilon$, we consider the Banach space
     $$
     {UC}^n_\epsilon=\left\{\phi\in C((-\infty,0];\er^n): \sup_{s\leq 0}\frac{|\phi(s)|}{\Ne^{-\epsilon s}}<+\infty,\, \frac{\phi(s)}{\Ne^{-\epsilon s}} \text{ is uniformly continuous on }(-\infty,0]\right\},
     $$
      equipped with the norm $\dst\|\phi\|_\epsilon=       \dst \sup_{s\leq 0}\frac{|\phi(s)|}{\Ne^{-\epsilon s}}$.

     In \cite{hino2006functional}, a basic theory about the existence, uniqueness, and continuation solutions  is established for the general functional differential equation in the phase space $UC_\epsilon^n$  
      \begin{align}\label{fde}
     	x'(t)=f(t,x_t),\Es t\geq 0,
     \end{align}
     where, for an open set $D\subseteq UC_\epsilon^n$, the function $f:[0,+\infty)\times D\to\er^n$ is continuous and $x_t$ denotes the function
    $x_t:(-\infty,0]\to\er^n$ defined by $x_t(s)=x(t+s)$ for $s\leq 0$.

    We denote by  $x(t,t_0,\phi)$ a solution of \eqref{fde} with initial condition $x_{t_0}=\phi \text{ for } t_0\geq 0 \text{ and } \phi\in D$.
    
   For $x\in\er^n$, we also use $x$ to denote the constant function $\phi(s)=x$ in $UC_\epsilon^n$. A vector $x\in \er^n$ is said to be positive if $x_i>0$ for all $i=1,\dots,n$ and we denote it by $x>0$.

   Now, we introduce the Banach space $BC$ of all continuous bounded functions $\phi:(-\infty,0]\to\er^n$ equipped with the norm $\|\phi\|=\dst\sup_{s\leq  0}|\phi(s)|$. It is clear that $BC\subseteq {UC}^n_\epsilon$ and we have $\|\phi\|_\epsilon\leq \|\phi\|$ for all $\phi\in BC$.

  In the phase space $UC_\epsilon^n$, for $n\in\en$ and $\epsilon>0$, we consider the following general nonautonomous differential system with infinite delays,
  \begin{align}\label{1.1}
  	x'_i(t)=a_i(t,x_i(t))\big[-b_i(t,x_i(t))+f_i(t,x_t)\big],\Es t\geq 0, \, i=1,\dots,n,
  \end{align}
  where $a_i:[0,+\infty)\times \er\to (0,\infty)$, $b_i:[0,+\infty)\times \er\to \er$, and $f_i:[0,+\infty)\times {UC}^n_\epsilon\to\er$ are continuous functions.

  The goal is to apply the results to Cohen-Grossberg neural network-type models, thus we only consider bounded initial conditions. i.e.
  \begin{align}\label{1.1 ic}
  	x_{t_0}&=\phi, \Es\text{ for } \phi \in BC \text{ and }  t_0\geq0.
  \end{align}

 The continuity of $a_i$, $b_i$, and $f_i$  functions assures that the initial value problem \eqref{1.1}-\eqref{1.1 ic} has a solution (see \cite[Theorem 2.1]{hale1978phase}).

    As we always consider bounded initial conditions, in this paper we consider the following definition of global exponential stability.
\begin{definicao}\label{defn 2.1}
	The system \eqref{1.1} is said to be globally exponentially stable if there are $\delta>0$ and $C\geq 1$ such that
	\begin{align*}
		|x(t,t_0,\phi)-x(t,t_0,\psi)|\leq C \Ne^{-\delta(t-t_0)}\|\phi-\psi\|,\Es\forall t_0\geq 0,\, \forall t\geq t_0,\, \forall \phi,\psi\in BC.
	\end{align*}
\end{definicao}
It should be emphasized that the preceding definition of global exponential stability is the usually used one in the literature on neural networks (\cite{zhang2009global},\cite{zhao2008dynamics},\cite{zhao2008dynamics}).

	\section{Global exponential stability}\label{section 2}
	
In this section, we obtain sufficient conditions for the global exponential stability of \eqref{1.1}. To do that in this section we assume the following hypotheses.

 For each  $i=1,\dots,n$:
 \begin{enumerate} [(H1)]
    \item there are $\un{a}_i, \ov{a}_i>0$ such that
    \begin{align*}
        \un{a}_i < a_i(t,u)<\ov{a}_i, \Es  \forall t\geq 0,\, \forall u\in\er;
    \end{align*}
    \item there exists a continuous function $D_i:[0,+\infty)\to \er$ such that
    \begin{align*}
        D_i(t)a_i^2(t,u)\leq \frac{\partial a_i}{\partial t}(t,u),\Es \forall t>0,\, \forall u\in \er;
    \end{align*}
    \item there exists a function $\beta_i:[0,+\infty)\to(0,+\infty)$ such that
    \begin{align*}
        \frac{b_i(t,u)-b_i(t,v)}{u-v}\geq \beta_i(t), \Es \forall t\geq 0, \, \forall u,v\in \er,\, u\neq v;
    \end{align*}
    \item the function $f_i:[0,+\infty)\times {UC}^n_\epsilon\to\er$ is  Lipschitz on its second variable
    i.e., there is a continuous function $\mathcal{L}_i:[0,+\infty)\to [0,+\infty)$ such that
    \begin{align*}
        |f_i(t,\phi)-f_i(t,\psi)|\leq \mathcal{L}_i(t)||\phi-\psi||_\epsilon, \Es \forall t\geq 0,\, \forall \phi,\psi\in {UC}^n_\epsilon;
    \end{align*}
    \item for all $t\geq 0$,
    \begin{align}\label{H5}
        \un{a}_i\big(\beta_i(t)+D_i(t)\big)-\ov{a}_i\mathcal{L}_i(t) >\epsilon.
    \end{align}
\end{enumerate}

By the generalized Gronwall's inequality \cite[Lemma 6.2]{hale1980ordinary} and the Continuation Theorem \cite[Theorem 2.4]{hale1978phase}, we can assure that the solutions of the initial value problem \eqref{1.1}-\eqref{1.1 ic} are defined on $\er$.

Now, we are in a position the obtain the main stability criterion for system \eqref{1.1}.
\begin{teorema}\label{theorem 2.2}
If (H1)-(H5) hold, then the system \eqref{1.1} is globally exponentially stable.
\end{teorema}
\begin{proof}
   Let $t_0>0$, $\phi=(\phi_1,\dots,\phi_n)\in BC$, $\psi=(\psi_1,\dots,\psi_n)\in BC$, and consider  two solutions, $x(t)=x(t,t_0,\phi)$ and $y(t)=x(t,t_0,\psi)$, of \eqref{1.1}.

   For each $t\geq  t_0$, define $V(t)=V(t,t_0,x(\cdot),y(\cdot))=(V_i(t),\ldots,V_n(t))\in\er^n$ by
    \begin{equation}\label{1.3}
        V_i(t):= \Ne^{\epsilon(t-t_0)}sign\big(x_i(t)-y_i(t)\big)\int_{y_i(t)}^{x_i(t)}\frac{1}{a_i(t,u)}du,\Es i=1,\dots,n.
    \end{equation}
From (H1), we conclude that
\begin{eqnarray}\label{est_V}
	\Ne^{-\epsilon(t-t_0)}\underline{a}_iV_i(t)\leq|x_i(t)-y_i(t)|\leq\Ne^{-\epsilon(t-t_0)}\overline{a}_iV_i(t),\Es\forall t\geq t_0,\,i=1,\ldots,n.
\end{eqnarray}

    Firstly, we show that
    \begin{equation}\label{1.4}
        |V(t)|\leq \max_i\{\un{a}_i^{-1}\}\|\phi-\psi\|,\Es \forall t\geq t_0.
    \end{equation}
    Obviously, from \eqref{est_V}, we have
    \begin{eqnarray*}
        |V(t_0)|\leq \max_i\big\{\un{a}_i^{-1} |x_i(t_0)-y_i(t_0)|\big\}\leq \max_i\{\un{a}_i^{-1}\}\|\phi-\psi\|.
    \end{eqnarray*}
    Now, to obtain a contradiction, we assume that inequality \eqref{1.4} is false. Consequently, there exists $t_1>t_0$ such that
    $$
    |V(t_1)|> \max_i\{\un{a}_i^{-1}\}\|\phi-\psi\|.
    $$
    Define
    $$
      T:=\min\left\{t\in[t_0,t_1]: V(t)=\max_{s\in[t_0,t_1]}|V(s)|\right\}.
    $$
    Choosing $i\in\{1,\dots,n\}$ such that $V_i(T)=|V(T)|$, we have
    \begin{eqnarray}\label{prop_V}
    	V_i(T)>0,\Es V'_i(T)\geq 0,\Es\text{and}\Es V_i(T)>|V(t)|,\,\, \forall t<T.
    \end{eqnarray}
   From \eqref{1.1}, and (H2), (H3), and (H4), we obtain
    \begin{eqnarray*}
        V'_i(T)&=&\epsilon V_i(T)+\Ne^{\epsilon(T-t_0)}sign\big(x_i(T)-y_i(T)\big)\left[\frac{1}{a_i(T,x_i(T))}x'_i(T)\right.\\
        & &\left.-\frac{1}{a_i(T,y_i(T))}y'_i(T) +\int_{y_i(T)}^{x_i(T)}-\frac{\frac{\partial a_i}{\partial t}(T,u)}{a_i^2(T,u)}du\right]\\
        &=&\epsilon V_i(T)+\Ne^{\epsilon(T-t_0)}sign\big(x_i(T)-y_i(T)\big)\bigg[b_i(T,y_i(T))-b_i(T,x_i(T))\\
        & &+f_i(T,x_T)-f_i(T,y_T)+\int_{y_i(T)}^{x_i(T)}-\frac{\partial_t a_i(T,u)}{a_i^2(T,u)}du\bigg]\\
         &\leq&\epsilon V_i(T)+\Ne^{\epsilon(T-t_0)}\big[-\beta_i(T) |x_i(T)-y_i(T))|+\mathcal{L}_i(T)||x_T-y_T||_\epsilon\\
        & &-D_i(T)|x_i(T)-y_i(T)|\big].
    \end{eqnarray*}
    Hypothesis (H5) implies $\beta_i(T)+D_i(T)>0$, and from \eqref{est_V}, we obtain
    \begin{eqnarray}\label{1.5}
        V'_i(T) &\leq& \epsilon V_i(T)-\un{a}_i \big[\beta_i(T)+D_i(T)\big]V_i(T)\nonumber\\
            & & +\Ne^{\epsilon(T-t_0)}\mathcal{L}_i(T)\max\left\{\sup_{s\leq t_0-T}|x(T+s)-y(T+s)|\Ne^{\epsilon s}, \sup_{t_0-T<s\leq 0}|x(T+s)-y(T+s)|\Ne^{\epsilon s}\right\}\nonumber\\
        &\leq& \epsilon V_i(T)-\un{a}_i \big[\beta_i(T)+D_i(T)\big]V_i(T)\nonumber\\
           & & + \Ne^{\epsilon(T-t_0)}\mathcal{L}_i(T)\max\big\{\|\phi-\psi\|\Ne^{\epsilon (t_0-T)}, \sup_{t_0-T<s\leq 0}|x(T+s)-y(T+s)|e^{\epsilon s}\big\}.\nonumber
    \end{eqnarray}
    By \eqref{est_V}, we obtain
    \begin{eqnarray*}
                V'_i(T) &\leq& \epsilon V_i(T)-\un{a}_i \big[\beta_i(T)+D_i(T)\big]V_i(T)\\
                 & & +\Ne^{\epsilon(T-t_0)}\mathcal{L}_i(T)\max\big\{\|\phi-\psi\|\Ne^{\epsilon (t_0-T)}, \sup_{t_0-T<s\leq 0}\Ne^{-\epsilon(T+s-t_0)+\epsilon s}\ov{a}_i V_i(T+s)\big\}\\
                 &=& \epsilon V_i(T)-\un{a}_i \big[\beta_i(T)+D_i(T)\big]V_i(T)+\ov{a}_i\mathcal{L}_i(T)\max\left\{\frac{\|\phi-\psi\|}{\ov{a}_i}, \sup_{t_0-T<s\leq 0} V_i(T+s)\right\}.
    \end{eqnarray*}
    By (H1), the definition of $T$, and \eqref{prop_V},  we have %$V_i(T)>\max_i\{\ov{a}_i^{-1}\}||\phi-\psi||$, and from  we have
    $$
     V'_i(T) \leq  \epsilon V_i(T)-\un{a}_i \big[\beta_i(T)+D_i(T)\big]V_i(T)+\ov{a}_i\mathcal{L}_i(T)V_i(T).
    $$
    From \eqref{prop_V} and (H5), we conclude that
    $$
       V'_i(T) \leq \left[\epsilon -\un{a}_i \big(\beta_i(T)+D_i(T)\big)+\ov{a}_i\mathcal{L}_i(T)\right]V_i(T)<0,
    $$
    which contradicts \eqref{prop_V} and hence \eqref{1.4} holds.

    From \eqref{est_V} and \eqref{1.4}, we obtain
    $$
      |x(t)-y(t)|\Ne^{\epsilon(t-t_0)}\min\left\{\ov{a}_i^{-1}\right\}\leq|V(t)|\leq\max_i\{\un{a}_i^{-1}\}\|\phi-\psi\|,
     $$
    thus
    $$
            |x(t)-y(t)|\leq C e^{-\epsilon(t-t_0)}\|\phi-\psi\|, \Es \forall t\geq t_0,
   $$
    with $C=\frac{\max_i\{\un{a}_i^{-1}\}}{\min_i\{\ov{a}_i^{-1}\}}=\frac{\max_i\{\ov{a}_i\}}{\min_i\{\un{a}_i\}}\geq1$, which shows that the system \eqref{1.1} is globally exponentially stable.
\end{proof}

We remark that hypothesis (H2) trivially holds (with $D_i(t)=0$ for all $t>0$) in case of all functions $a_i$ do not explicitly depend on time $t$, i.e. $a_i(t,u)=a_i(u)$ for all $i=1,\ldots,n$ and $u\in\er$. Thus, under the assumption
	\begin{description}
		\item[(h5)] For all $t\geq0$ and $i=1,\ldots,n$, we have $\un{a}_i\beta_i(t)-\ov{a}_i\mathcal{L}_i(t)>\epsilon$,
	\end{description}
we have the following result for system
	\begin{eqnarray}\label{eq:modelo-funcao-amp-ind-t}
	  	x'_i(t)=a_i(x_i(t))\big[-b_i(t,x_i(t))+f_i(t,x_t)\big],\Es t\geq 0, \, i=1,\dots,n.
	\end{eqnarray}
\begin{corolario}\label{cor:sem-hipotese-H2}
	Assume (H1), (H3), (H4), and (h5) hold. Then, system~\eqref{eq:modelo-funcao-amp-ind-t} is globally exponentially stable.
\end{corolario}
\begin{proof}
	Hypothesis (H2) holds with $D(t)=0$, thus the result comes from Theorem~\ref{theorem 2.2}.
\end{proof}	

Now consider the model studied in \cite{oliveira2017global}
\begin{eqnarray}\label{model_jj_NN_2017}
	x'(t)=a_i(t,x_i(t))\left[-b_i(t,x_i(t))+\sum_{k=1}^K\sum_{j=1}^nf_{ijk}(t,{x_j}_t)\right],\Es t\geq0,\,i=1,\ldots,n,
\end{eqnarray}
where $n,K\in\en$, $a_i$ and $b_i$ are functions as in system \eqref{1.1} and $f_{ijk}:[0,+\infty)\times UC_\epsilon^1\to\er$ are continuous functions for $i,j=1,\ldots,n$ and $k=1,\ldots,K$.

We will also assume the following conditions:
	\begin{description}
		\item[(h4)] for each $i,j=1,\ldots,n$ and $k=1,\ldots,K$, there exists a continuous function $\mathcal{F}_{ijk}:[0,+\infty)\to[0,+\infty)$ such that
		$$
		  |f_{ijk}(t,\varphi)-f_{ijk}(t,\psi)|\leq \mathcal{F}_{ijk}(t)\|\varphi-\psi\|_\epsilon,\Es\forall t\geq0,\,\varphi,\phi\in UC^1_\epsilon.
		$$
		\item[(h5')] for all $t\ge 0$ and $i=1,\ldots,n$, we have
        $$\un{a}_i\big(\beta_i(t)+D_i(t)\big)-\ov{a}_i\sum_{k=1}^K\sum_{j=1}^n\mathcal{F}_{ijk}(t)>\epsilon.$$
	\end{description}

As system \eqref{model_jj_NN_2017} is a particular situation of \eqref{1.1}, the following stability criterion holds.
\begin{corolario}\label{cor:model-jj17}
	Assume that (H1), (H2), (H3), (h4) and (h5') hold. Then system \eqref{model_jj_NN_2017} is globally exponentially stable.
\end{corolario}
\begin{proof}
  System \eqref{model_jj_NN_2017} is a particular situation of \eqref{1.1} with
  $$
    f_i(t,\varphi)=\sum_{k=1}^K\sum_{j=1}^nf_{ijk}(t,\varphi_j),\Es\forall t\geq0,\,\varphi=(\varphi_1,\ldots,\varphi_n)\in UC_\epsilon^n.
  $$
  From (h4), we know that (H4) holds with
  $$
    \mathcal{L}_i(t)=\sum_{k=1}^K\sum_{j=1}^n\mathcal{F}_{ijk}(t),\Es\forall t\geq0,\,i=1,\ldots,n.
  $$
  Moreover, (H5) reads as (h5'). Thus the results comes from Theorem~\ref{theorem 2.2}.
\end{proof}

\begin{rem}\label{rem:criterio-estabilidade-jjNN17} We remark that the exponential stability of \eqref{model_jj_NN_2017} was proved in \cite{oliveira2017global} under the assumptions (H1), (H2), (H3), (h4), and a condition equivalent to
\begin{description}
    \item[(h5'')] for all $t\ge 0$ and $i=1,\ldots,n$, we have
    \begin{eqnarray}\label{cond-H5-em-jj17}
	   \un{a}_i\big(\beta_i(t)+D_i(t)\big)-\sum_{k=1}^K\sum_{j=1}^n\ov{a}_j\mathcal{F}_{ijk}(t)>\epsilon.
    \end{eqnarray}
	\end{description}
We emphasize that conditions (h5') and \eqref{cond-H5-em-jj17} are different, thus Corollary \ref{cor:model-jj17} presents a new exponential stability criterion for the system \eqref{model_jj_NN_2017}.
\end{rem}

\section{Existence of periodic solution}\label{not+basicResults}
In this section, we assume that \eqref{1.1} is a periodic system and we establish sufficient conditions for the existence of a periodic solutions.

The existence of a periodic solutions will be proved through Mawhin's Continuation Theorem. Before stating the referred theorem, we need to recall some definitions and facts.

\begin{definicao}
	Let $X$ and $Z$ two Banach spaces.\\	
A linear mapping $L: \Dom\, L \subseteq X \to Z$ is called a \emph{Fredholm mapping of index zero} if $\dim \Ker_L = \codim \Imagem_L < \infty$ and $\Imagem_L$ is closed in $Z$.
\end{definicao}
Given a Fredholm mapping of index zero, $L: \Dom\, L \subseteq X \to Z$ , it is well known that there are continuous projectors $P:X\to X$ and $Q:Z\to Z$ such that
$\Imagem_P = \Ker_L$, $\Ker_Q = \Imagem_L = \Imagem_{I-Q}$, $X = \Ker_L \oplus \Ker_P$ and $Z = \Imagem_L \oplus \Imagem_Q$.
It follows that $L|_{\Dom\, L \,\cap\, \ker_P}: \Dom\, L \,\cap \,\ker_P \to \Imagem_L$ is invertible. We denote the inverse of that map by $K_P$.
\begin{definicao}
Let $U$ be an open bounded subset of $X$. We say that a continuous mapping $N: \overline{U} \subseteq X \to Z$ is $L$-compact on $\overline{U}$ if the set $QN(\overline U)$ is bounded and the mapping $K_P(I-Q)N: \overline{U} \subseteq X \to X$ is compact.
\end{definicao}

\begin{teorema}[Mawhin's Continuation Theorem]\label{mwahin}
	Let $X$ be a Banach space and $\Omega\subseteq X$ an open bounded set. Suppose $L:\Dom\, L\subset X\rightarrow X$ is a Fredholm operator with zero index and that $N:\overline{\Omega}\rightarrow X$ is L-compact on $\overline{\Omega}$ . Moreover, assume that all the following conditions are satisfied:
	\begin{enumerate}[1.]
		\item $ Lx\neq\lambda Nx,\quad\forall x \in \partial\Omega\cap \Dom\, L,\, \lambda\in(0,1)$;
		\item $QNx\neq 0,\quad\forall x \in \partial\Omega\cap \Ker\, L$;
		\item ${\deg}_B\{QN, \Omega\cap  \Ker\, L,0\}\neq 0$, where  ${\deg}_B$ denotes the Brouwer degree.
	\end{enumerate}
	Then, the equation $Lx=Nx $ has at least one solution in $\overline{\Omega}$.
\end{teorema}

For studying the system \eqref{1.1} in case of being periodic, the following hypotheses will be considered:
 \begin{description}
 	\item(H1*) For each $i=1,\ldots,n$, there exist $\ov a_i\, , \un a_i>0$ such that 
  \begin{eqnarray*}
            \un a_i< a_i(t,u)< \ov a_i \text{ for all  } t\geq0,\, u\in\er;
  \end{eqnarray*}%\label{eq-bound-ai-periodic}
     % \un a_i< a_i(t,u)< \ov a_i \text{ for all  } t\geq0,\, u\in\er;
 	%$\un a_i< a_i(t,u)< \ov a_i$ for all $t\geq0$, $u\in\er$.
 	\item[(H2*)] There is $\omega>0$ such that, for each $i=1,\ldots,n$,
 	$$
 	  a_i(t,u)=a_i(t+\omega,u),\Es b_i(t,u)=b_i(t+\omega,u),\Es f_i(t,\phi)=f_i(t+\omega,\phi)
 	$$
 	for all $t\geq0$, $u\in\er$, and $\phi\in BC$;
	\item[(H3*)] For each $i=1,\dots,n$, there exist $\omega-$periodic continuous functions $\beta_i,\beta_i^*:[0,+\infty)\to(0,+\infty)$ such that
	\begin{align*}
		\beta_i(t)\leq\frac{b_i(t,u)-b_i(t,v)}{u-v}\leq \beta_i^*(t), \Es \forall t\in[0,\omega], \, \forall u,v\in \er,\, u\neq v;
	\end{align*}
 \item[(H4*)] For each $i=1,\dots,n$, there exists a $\omega-$periodic continuous function $\mathcal{L}_i:[0,+\infty)\to [0,+\infty)$ such that
\begin{align*}
	|f_i(t,\phi)-f_i(t,\psi)|\leq \mathcal{L}_i(t)\|\phi-\psi\|, \Es \forall t\in[0,\omega],\, \forall \phi,\psi\in BC;
\end{align*}
	\item[(H5*)] For each $i=1,\dots,n$,
	\begin{align*}
		\beta_i(t)>\mathcal{L}_i(t),\Es \forall t\in[0,\omega].
	\end{align*}
\end{description}

%{\color{blue}{We remark that, by the periodicity assumed in (H2*), we can consider all function, $a_i$, $b_i$, and $f_i$ with $t\in\er$ " are bounded".}}

From (H2*), we conclude that the continuous functions $t\mapsto b_i(t,0)$ and $t\mapsto f_i(t,0)$ are $\omega$-periodic and therefore bounded. From (H3*), we also conclude that $\beta_i$ are bounded away from zero and $\beta^*_i$ are bounded.

Defining
\begin{eqnarray}\label{def-sup-coef}
  \un{\beta}_i:=\min_{t\in[0,\omega]}\beta_i(t),\,\,\,\ov{\beta}^*_i:=\max_{t\in[0,\omega]}\beta^*_i(t),\,\,\,\ov{b}_i:=\max_{t\in[0,\omega]}|b_i(t,0)|,\,\,\text{and}\,\,\,\ov{f}_i:=\max_{t\in[0,\omega]}|f_i(t,0)|,
\end{eqnarray}
so that we have $0< \un{\beta}_i, \ov{\beta}^*_i$, and $0\leq \ov{b}_i,\ov{f}_i$.

We denote by $X$ the Banach space
$$
  X=\big\{\phi\in C(\er:\er^n): \phi\text{ is }\omega-\text{periodic}\},
$$
with the norm $\|\phi\|=\dst\sup_{t\in[0,\omega]}|\phi(t)|$, for $\phi\in X$.

For $\Dom_L=\{\phi\in X: \phi'\in X\}\subseteq X$, define the linear operator $L: \Dom_L\to X$ by
\begin{eqnarray}\label{3.4}
  L\phi=\phi'
\end{eqnarray}
i.e., for all $t\in\er$ and $\phi(t)=(\phi_1(t),\ldots,\phi_n(t))\in \Dom_L$, we have $\big(L\phi\big)(t)=(\phi_1'(t),\ldots,\phi_n'(t))$.

It is not difficult to show that $\Ker_L\cong\er^n$ and
\begin{eqnarray}\label{def-ImL}
  \Imagem_L=\left\{\phi=(\phi_1,\ldots,\phi_n)\in X:\int_0^\omega\phi_1(t)dt=\dots=\int_0^\omega\phi_n(t)dt=0\right\},
\end{eqnarray}
with $\Imagem_L$ closed in $X$ and $\dim \Ker_L=\codim \Imagem_L=n$, thus $L$ is a Fredhom operator with zero index.

Now, we consider the projection $P:X\to X$ defined by
\begin{eqnarray}\label{def-P}
  P\phi=\frac{1}{\omega}\int_0^\omega\phi(t)dt=\frac{1}{\omega}\left(\int_0^\omega\phi_1(t)dt,\dots,\int_0^\omega\phi_n(t)dt\right),\Es\forall\phi=(\phi_1,\ldots,\phi_n)\in X.
\end{eqnarray}

The projection $P$ is continuous and, considering $Q\phi=P\phi$, we have $\Imagem_P=\Ker_L$, $\Ker_Q=\Imagem_L$, and the operator $L_{|_{\Dom_L\cap \Ker_P}}:\Dom_L\cap \Ker_P\to \Imagem_L$ is invertible and we denote the inverse by $K_P$. By \eqref{3.4} and \eqref{def-ImL}, we obtain that $K_p\phi=\big((K_P\phi)_1,\cdots,(K_P\phi)_n\big)$ with
\begin{eqnarray}\label{3.5'}
 (K_P\phi)_i(t)=\int_0^t\phi_i(u)du-\frac{1}{\omega}\int_0^\omega\int_0^u\phi_i(s)dsdu,\,\,\forall \phi=(\phi_1,\ldots,\phi_n)\in \Imagem_L,
\end{eqnarray}
for $i=1,\ldots,n$.

For a convenient bounded open set $\Omega\subseteq X$, define the function $N:\ov{\Omega}\to X$ by $N\phi=\big((N\phi)_1,\ldots,(N\phi)_n\big)$, where
\begin{eqnarray} \label{3.3}
	(N\phi)_i(t)=a_i(t,\phi_i(t))\bigg[-b_i(t,\phi_i(t))+f_i(t,\phi_t)\bigg],
\end{eqnarray}
for all $t\in\er$, $\phi=(\phi_1,\ldots,\phi_n)\in X$, and $i=1,\ldots,n$.

We claim that, from the continuity of $a_i$, $b_i$, and $f_i$, \eqref{3.5'} and \eqref{3.3}, we can conclude that, for any $\alpha>0$, the mapping $N$ is $L$-compact in the set $\Omega=\{\phi \in X:\|\phi\|<\alpha\}$.

In fact, for any $t \in \er$ and any $x \in X$, we have $\displaystyle\|QNx\|\le \max_{i} \ov{a}_i[2\ov{\beta}_i^*\alpha+\ov{b}_i+\ov{f}_i]$, and we conclude that $QN(X)$ is bounded, implying that $QN(\overline{\Omega})$ is bounded.

Additionally, we also need to show that the mapping $K_P(I-Q)N$ is compact. To achieve this, we show that for any bounded $V \subseteq \overline{\Omega}$, the set $\overline{K_P(I-Q)N(V)}$ is compact. It is easy to verify that, for any sequence, $(\phi_n)$, with $\phi_n \in V$, $n \in \mathbb{N}$, such that $\phi_n \to \phi$, we have, for any $t,t_0 \in \er$,
\begin{equation}\label{eq:family-equicont}
\begin{split}
&\lim_{n \to +\infty} \left|K_P(I-Q)N(\phi_n)(t)-K_P(I-Q)N(\phi_n)(t_0)\right|\\
& \quad \quad \quad \le 3 \max_{i} [\ov{a}_i(2\ov{\beta}_i^*\alpha+\ov{b}_i+\ov{f}_i)] \, (t-t_0).
\end{split}
\end{equation}
%\begin{equation}\label{eq:family-equicont}
%\begin{split}
%&\lim_{n \to +\infty} \left|K_P(I-Q)N(\phi_n)(t)-K_P(I-Q)N(\phi_n)(t_0)\right|\\
%& \quad \uqad \quad \le 3 \max_{i} [\ov{a}_i(\ov{b}_i+\ov{f}_i)] \, (t-t_0).
%\end{split}
%\end{equation}
and
\begin{equation}\label{eq:family-bounded}
\lim_{n \to +\infty} \left\|K_P(I-Q)N(\phi_n)\right\| \le 3 \omega \max_{i} [\ov{a}_i(2\ov{\beta}_i^*\alpha+\ov{b}_i+\ov{f}_i)].
\end{equation}
%\begin{equation}\label{eq:family-bounded}
%\lim_{n \to +\infty} \left\|K_P(I-Q)N(\phi_n)\right\| \le 3 \omega \max_{i} [\ov{a}_i(\ov{b}_i+\ov{f}_i)].
%\end{equation}
Inequality~\eqref{eq:family-equicont} shows that the family of functions $\overline{K_P(I-Q)N(V)}$ is equicontinuous and inequality
\eqref{eq:family-bounded} shows that the norms of all the functions in the referred family of functions are bounded by the same constant. Ascoli-Arzela theorem allows us to conclude that the set $\overline{K_P(I-Q)N(V)}$ is compact. Thus the
mapping $K_P(I-Q)N$ is compact and the claim is proved.

Notice that equation~\eqref{eq:family-bounded} only allows us to conclude that
$$\displaystyle \lim_{n \to +\infty} \left|K_P(I-Q)N(\phi_n)(t)\right| \le 3 \omega \max_{i} [\ov{a}_i(\ov{b}_i+\ov{f}_i)], \text{ for any } t \in [0,\omega].$$ Thus we are not able to apply directly Ascoli-Arzela's theorem to functions in $$\ov{K_P(I-Q)N(\ov{\Omega})}.$$ Instead, we must consider the space
$\widetilde{\Omega}=\{\phi \in C([0,\omega]:\er^n):\|\phi\|<\alpha\}$ instead of $\Omega$,
with the norm defined in the same way. This is not a problem since once we show the compactness property for $\widetilde{\Omega}$, the same property holds for $\Omega$, because the functions on $\Omega$ are $\omega-$periodic.

In view of \eqref{3.3} and \eqref{3.4}, for $\lambda\in(0,1)$ and $x(t)=(x_1(t),\ldots,x_n(t))\in X$, the operator equation $Lx=\lambda Nx$ is equivalent to the following equation:
\begin{eqnarray} \label{3.5}
	x'_i(t)=\lambda a_i(t,x_i(t))\bigg[-b_i(t,x_i(t))+f_i(t,x_t)\bigg], \Es \forall\lambda\in(0,1),\, i=1,\ldots,n.
\end{eqnarray}

Now we are in a position to prove the existence of a periodic solution of the general differential system \eqref{1.1}.
\begin{teorema}\label{teo:existencia-orbita-periodica}
	 Suppose that {(H1*)}, (H2*), (H3*), (H4*), and (H5*) hold.
	Then, system \eqref{1.1} has at least one $\omega-$periodic solution.
\end{teorema}
\begin{proof}
	Our objective is to apply Theorem \ref{mwahin}. To accomplish this, it is needed to define a bounded open set $\Omega\subseteq X$ for which the conditions 1., 2., and 3. in Theorem \ref{mwahin} hold.
	
	Let $x=x(t)=(x_1(t),\dots,x_n(t))^T$ be an arbitrary $\omega-$periodic solution of equation \eqref{3.5}. The components $x_i(t)$ of $x(t)$ are all continuously differentiable, thus, for each $i=1,\ldots,n$, there is $t_i\in [0,\omega]$ such that
	$$
	  |x_i(t_i)|= \max_{t \in  [0,\omega]}|x_i(t)|.
	$$
	Hence $x'_i(t_i)=0$ for all $i=1,\dots,n$.
	
	Choose $i\in\{1,\ldots,n\}$ such that $|x_i(t_i)|=\dst\max_{t\in[0,\omega]}|x(t)|$.
	Consequently, from  \eqref{3.5}, we have
	\begin{eqnarray} \label{3.6}
		b_i(t_i,x_i(t_i))=f_i(t_i,x_{t_i}),
	\end{eqnarray}
    thus
    \begin{eqnarray*}
    	b_i(t_i,x_i(t_i))-b_i(t_i,0)+b_i(t_i,0)=f_i(t_i,x_{t_i})-f_i(t_i,0)+f_i(t_i,0).
    \end{eqnarray*}
    By (H3*), (H4*), and \eqref{def-sup-coef} we obtain
     \begin{eqnarray*}
    	\beta_i(t_i)|x_i(t_i)|-\ov{b}_i\leq\mathcal{L}_i(t_i)\|x_{t_i}\|+\ov{f}_i,
    \end{eqnarray*}
    and, as $\|x_{t_i}\|=|x(t_i)|=|x_i(t_i)|$, we get
    $$
      |x_i(t_i)|\left(1-\frac{\mathcal{L}_i(t_i)}{\beta_i(t_i)}\right)\leq\frac{\ov{f}_i+\ov{b}_i}{\beta_i(t_i)},
    $$
    From (H2*), (H5*), and \eqref{def-sup-coef}, we can define
    \begin{eqnarray}\label{def-xi}
      \ov\xi=\max_{j,t}\left\{\left(1-\frac{\mathcal{L}_j(t)}{\beta_j(t)}\right)^{-1}\frac{\ov{f}+\ov{b}}{\un{\beta}}\right\}+1>0,
    \end{eqnarray}
    where $\ov{b}=\dst\max_i\ov{b}_i$, $\ov{f}=\dst\max_i\ov{f}_i$, and $\un{\beta}=\dst\min_i\un{\beta}_i$, thus we conclude that
    \begin{eqnarray}\label{ineq-con-1}
    |x_i(t_i)|<\ov\xi.
   \end{eqnarray}
    Consequently, $\|x\|<\ov\xi$, and taking
    \begin{eqnarray}\label{def-omega}
      \Omega=\big\{\phi\in X: \|\phi\|<\ov\xi\big\},
    \end{eqnarray}
    we conclude that the first condition of Theorem \ref{mwahin} is satisfied.

    Now, we prove that the second condition of Theorem \ref{mwahin} holds.

    Let $x=x(t)=(x_1(t),\dots,x_n(t))^T\in\partial \Omega\cap \Ker_L$. As  $\Ker_L\cong\er^n$, then $x(t)$ is a constant vector in $\er^n$, i.e. $x(t)=(x_1,\ldots,x_n)$, and by \eqref{def-omega}, we conclude that there is $i\in\{1,\dots,n\}$ such that $|x_i|=\ov{\xi}$. By \eqref{def-P} and \eqref{3.3}, we have
    \begin{eqnarray*}
    	(QNx)_i(t)=(QNx)_i=\frac{1}{\omega}\int_{0}^{\omega}a_i(u,x_i)\left[-b_i(u,x_i)+f_i(u,x)\right]du.
    \end{eqnarray*}
    We claim that \begin{eqnarray}\label{3.12}
    	|(QNx)_i|>0.
    \end{eqnarray}
    By contradiction, we assume that $|(QNx)_i|=0.$ Then there is $t^*_i\in [0,\omega]$ such that
    $$
      b_i(t_i^*,x_i)=f_i(t_i^*,x).
    $$
    Reproducing the same computations above (see how \eqref{3.6} implies \eqref{ineq-con-1}), we conclude that
    $$
      \ov\xi=|x_i|<\ov\xi,
    $$
    which is a contradiction. Consequently, \eqref{3.12} holds and the second condition of Theorem \ref{mwahin} is proved.

    In order to prove the last condition of Theorem \ref{mwahin}, we consider the continuous function $\Psi:\left(\Omega\cap \Ker_L\right)\times[0,1]\to X$ defined by $\Psi(x,\mu)=\left(\Psi(x,\mu)_1,\ldots,\Psi(x,\mu)_n\right)$ with
    $$
      \Psi(x,\mu)_i=-\mu\ov{a}_i\ov\beta_i^*x_i+(1-\mu)(QNx)_i,
    $$
    for all $x=(x_1,\ldots,x_n)\in\Omega\cap \Ker_L\cong\Omega\cap\er^n$, $\mu\in(0,1)$, and $i=1,\ldots,n$.
    We claim that
    \begin{eqnarray}\label{claim_no_zero}
    	|\Psi(x,\mu)|\neq0,\Es\forall x\in(\partial\Omega)\cap \Ker_L,\,\mu\in[0,1].
    \end{eqnarray}
    Consequently, defining $\Phi:\er^n\to\er^n$ by
    $$
      \Phi x=\left(-\ov{a}_1\ov{\beta}^*_1x_1,\ldots,-\ov{a}_n\ov{\beta}^*_nx_n\right),\Es\forall x=(x_1,\ldots,x_n)\in\er^n,
    $$
    the homotopy invariance theorem \cite{mawhin2005periodic} implies that
    $$
      \deg_B\left\{QN,\Omega\cap \Ker_L,0\right\}=\deg_B\left\{\Phi,\Omega\cap \Ker_L,0\right\}\neq0.
    $$
    Now, it remains to prove that \eqref{claim_no_zero} holds to conclude the proof.

    Let $x=(x_1,\ldots,x_n)\in (\partial\Omega)\cap \Ker_L$ and $\mu\in[0,1]$. The function $x$ is constant because $\Ker\cong\er^n$  and, by \eqref{def-omega}, we conclude that there  is $i\in\{1,\ldots,n\}$ such that $|x|=|x_i|=\ov\xi$. We claim that
    $$
    |  \Psi(x,\mu)_i|\neq0.
    $$
     By contradiction assume that
     \begin{eqnarray}\label{eq-abs}
     	|\Psi(x,\mu)_i|=0.
     \end{eqnarray}
     From  \eqref{def-P}, \eqref{3.3}, and \eqref{eq-abs}, we have
    $$
      -\mu\ov{a}_i\ov\beta_i^*x_i+\frac{1-\mu}{\omega}\int_0^\omega a_i(t,x_i)\big[-b_i(t,x_i)+f_i(t,x)\big]dt=0,
    $$
    thus there exists $t^{**}_i\in[0,\omega]$ such that
    \begin{eqnarray}\label{eq-degree-zero}
    -\mu\ov{a}_i\ov\beta_i^*x_i+(1-\mu) a_i(t^{**}_i,x_i)\big[-b_i(t^{**}_i,x_i)+f_i(t^{**}_i,x)\big]=0.
    \end{eqnarray}
Now, we assume that $|x|=x_i=\ov\xi>0$ (the situation $|x|=-x_i=\ov\xi$ is analogous).

{By condition (H1*) and (H3*),  we have}
\begin{align*}
	a_i(t^{**}_i,x_i)b_i(t^{**}_i,x_i)&=a_i(t^{**}_i,x_i)\big[b_i(t^{**}_i,x_i)-b_i(t^{**}_i,0)\big]+a_i(t^{**}_i,x_i)b_i(t^{**}_i,0)\\
	&\leq \ov{a}_i\ov\beta^*_ix_i+a_i(t^{**}_i,x_i)b_i(t^{**}_i,0),
\end{align*}
then
$$
a_i(t^{**}_i,x_i)b_i(t^{**}_i,x_i)- \ov{a}_i\ov{\beta}^*_ix_i-a_i(t^{**}_i,x_i)b_i(t^{**}_i,0)\leq 0.
$$
Consequently, from \eqref{eq-degree-zero}, we have
\begin{align*}
	-a_i&(t^{**}_i,x_i)b_i(t^{**}_i,x_i)+(1-\mu)a_i(t^{**}_i,x_i)f_i(t^{**}_i,x) \\
	\geq&\mu\left[a_i(t^{**}_i,x_i)b_i(t^{**}_i,x_i)- \ov{a}_i\ov{\beta}^*_ix_i-a_i(t^{**}_i,x_i)b_i(t^{**}_i,0)\right]-a_i(t^{**}_i,x_i)b_i(t^{**}_i,x_i)\\
	&+(1-\mu)a_i(t^{**}_i,x_i)f_i(t^{**}_i,x)\\
	=&-\mu\ov{a}_i\ov{\beta}^*_ix_i+(1-\mu)a_i(t^{**}_i,x_i)\left[-b_i(t^{**}_i,x_i)+f_i(t^{**}_i,x)\right]-\mu a_i(t^{**}_i,x_i)b_i(t^{**}_i,0)\\
	=&-\mu a_i(t^{**}_i,x_i)b_i(t^{**}_i,0)\\
	\geq&a_i(t^{**}_i,x_i)\min\big\{0,-b_i(t^{**}_i,0)\big\}
\end{align*}
and {by (H1*)}, we obtain
$$
  -b_i(t^{**}_i,x_i)+(1-\mu)f_i(t^{**}_i,x)\geq\min\big\{0,-b_i(t^{**}_i,0)\big\}.
$$
Consequently,
$$
b_i(t^{**}_i,x_i)-b_i(t^{**}_i,0)\leq|f_i(t^{**}_i,x)-f_i(t^{**}_i,0)|+\ov{b}_i+\ov{f}_i,
$$
 recalling that $x_i>0$, and $\|x\|=|x|$, from (H3*), (H4*), and \eqref{def-sup-coef} we have
$$
  x_i\leq\frac{\mathcal{L}_i(t^{**}_i)}{\beta_i(t^{**}_i)}|x|+\frac{\ov{b}_i+\ov{f}_i}{\un\beta_i}.
$$
As $|x|=x_i=\ov\xi>0$, we obtain
$$
  \ov{\xi}=x_i\leq\left(1-\frac{\mathcal{L}_i(t^{**}_i)}{\beta_i(t^{**}_i)}\right)^{-1}\frac{\ov{b}_i+\ov{f}_i}{\un\beta_i},
$$
and by \eqref{def-xi} we conclude {that}
$$
   \ov{\xi}=x_i\leq\left(1-\frac{\mathcal{L}_i(t^{**}_i)}{\beta_i(t^{**}_i)}\right)^{-1}\frac{\ov{b}_i+\ov{f}_i}{\un\beta_i}<\ov\xi,
$$
which is a contradiction.

{The case when $x_i<0$ is very similar to the previous one and we present it briefly. From (H1*), (H3*), and \eqref{def-sup-coef}, we obtain
$$
a_i(t^{**}_i,x_i)b_i(t^{**}_i,x_i)- \ov{a}_i\ov{\beta}^*_ix_i-a_i(t^{**}_i,x_i)b_i(t^{**}_i,0)\geq 0
$$
and, from {(H1*)}, (H5*), and \eqref{def-xi}, we obtain
$$
-b_i(t^{**}_i,x_i)+(1-\mu)f_i(t^{**}_i,x)\leq\max\big\{0,-b_i(t^{**}_i,0)\big\}.
$$
Therefore,
$$
b_i(t^{**}_i,x_i)-b_i(t^{**}_i,0)\geq-|f_i(t^{**}_i,x)-f_i(t^{**}_i,0)|-\ov{b}_i-\ov{f}_i.
$$
Since $x_i<0$ and $\|x\|_\epsilon=|x|$, from (H3*), (H4*), \eqref{def-sup-coef} and taking into account that $|x|=-x_i=-\ov\xi<0$, we obtain
$$
x_i\geq\frac{\mathcal{L}_i(t^{**}_i)}{\beta_i(t^{**}_i)}x_i-\frac{\ov{b}_i+\ov{f}_i}{\un\beta_i}.
$$
Using this last equation and \eqref{def-xi}, we conclude that
$$
-\ov{\xi}=x_i\geq-\left(1-\frac{\mathcal{L}_i(t^{**}_i)}{\beta_i(t^{**}_i)}\right)^{-1}\frac{\ov{b}_i+\ov{f}_i}{\un\beta_i}>-\ov\xi,
$$
and we obtain again a contradiction.}
\end{proof}

By the stability criteria established in the previous Section, now we are in a position to present the following results.

From Theorems \ref{theorem 2.2} and \ref{teo:existencia-orbita-periodica}, we have the following result.
\begin{teorema}\label{teo:conjunto-modelo-geral}
	Assume {(H1*), (H2*)}, (H2), (H3*), (H4) with $\mathcal{L}_i$ $\omega-$periodic continuous functions, (H5*), and (H5).
	Then the system \eqref{1.1} has an $\omega-$periodic solution which is globally exponentially stable.
\end{teorema}

In the case of $D_i(t)\leq0$, for all $t\geq0$ and $i=1,\ldots,n$, hypothesis (H5) implies (H5*), thus the following result is an immediate consequence of Theorem \ref{teo:conjunto-modelo-geral}.
\begin{corolario}\label{cor:conjunto-modelo-geral}
If {(H1*), (H2*)}, (H2) with $D_i(t)\leq0$, for all $t\geq0$ and $i=1,\ldots,n$, (H3*), (H4) with $\mathcal{L}_i$ $\omega-$periodic continuous functions, and (H5) hold, then the system \eqref{1.1} has an $\omega-$periodic solution which is globally exponentially stable.
\end{corolario}

In {the} particular case of functions $a_i$ {that} do not explicitly depend on time $t$, from the Corollary \ref{cor:conjunto-modelo-geral}, we have the following result.
\begin{corolario}
	If {(H1*)}, (H3*), (H4) with $\mathcal{L}_i$ $\omega-$periodic continuous functions, and (H5) hold, then system \eqref{eq:modelo-funcao-amp-ind-t} has an $\omega-$periodic solution which is globally exponentially stable.
\end{corolario}

Now, we assume that the system \eqref{model_jj_NN_2017} is $\omega-$periodic, i.e. the following hypothesis holds:
\begin{description}
	\item[{(h1*)}]  There is $\omega>0$ such that, for each $i,j=1,\ldots,n$ and $k=1,\ldots,K$,
	$$
	{a_i(t,u)=a_i(t+\omega,u)},\Es b_i(t,u)=b_i(t+\omega,u),\Es f_{ijk}(t,\phi)=f_{ijk}(t+\omega,\phi),
	$$
	for all $t\geq0$, $u\in\er$, and $\phi\in BC$.
\end{description}
From Corollary \ref{cor:model-jj17}, Remark \ref{rem:criterio-estabilidade-jjNN17}, and Theorem \ref{teo:existencia-orbita-periodica}, we obtain the next result.
\begin{teorema}
	Assume {(h1*)}, (H1*), (H2), (H3*), (h4) with $\mathcal{F}_{ijk}$ $\omega-$periodic continuous functions, and
	$$
	  \beta_i(t)>\sum_{k=1}^K\sum_{j=1}^n\mathcal{F}_{ijk}(t),\Es\forall t\in[0,\omega],\,i=1,\ldots,n.
	$$
	If one of the conditions {(h5') or (h5'')} holds, then the system \eqref{model_jj_NN_2017} has an $\omega-$periodic solution which is globally exponentially stable.
	
\end{teorema}

\section{Applications to Cohen-Grossberg neural network models}\label{aplicacoes}

In this section, we apply the results in Sections \ref{section 2} and \ref{not+basicResults} to Cohen-Grossberg type models. As we want to apply it to low-order and high-order models, we consider the following general Cohen-Grossberg model with discrete-time varying and distributed delays.
 \begin{eqnarray} \label{3.1}
	x'_i(t)&=&a_i(t,x_i(t))\bigg[-b_i(t,x_i(t))+F_i\bigg(\dst\sum_{p=1}^{P}\sum_{j,l=1}^{n}c_{ijlp}(t)h_{ijlp}\big(x_j(t-\tau_{ijp}(t)),x_l(t-{\til{\tau}}_{ilp}(t)\big)\bigg)\nonumber\\
	& &+G_i\bigg(\dst\sum_{q=1}^{Q}\sum_{j,l=1}^{n}d_{ijlq}(t)f_{ijlq}\left(\int_{-\infty}^0g_{ijq}(x_j(t+s))d\eta_{ijq}(s),\int_{-\infty}^0{\til{g}}_{ilq}(x_l(t+s))d{\til{\eta}}_{ilq}(s)\right)\bigg)\nonumber\\
	& &+I_i(t)\bigg],\Es t\geq 0,\Es \,i=1,\ldots,n,
\end{eqnarray}
where $n,P,Q\in\en$ and $a_i:[0,+\infty)\times\er\to(0,+\infty)$, $b_i:[0,+\infty)\times\er\to\er$,  $c_{ijlp}\,,d_{ijlq},I_i:[0,+\infty)\to\er$, $\tau_{ijp},\til{\tau}_{ilp}:[0,+\infty)\to[0,+\infty)$,  $h_{ijlp}\,,f_{ijlq}:\er^2\to\er$, $F_i, G_i, g_{ijq}, \til{g}_{ilq}:\er\to\er$ are continuous functions, and $\eta_{ijq},\,\til{\eta}_{ilq}:(-\infty,0]\to\er$ are non-decreasing bounded functions such that $\eta_{ijq}(0)-\eta_{ijq}(-\infty)=1$ and  $\til{\eta}_{ilq}(0)-\til{\eta}_{ilq}(-\infty)=1$, for each $i,j,l=1,\dots,n$, $p=1,\dots,P$, and $q=1,\dots,Q$.

Here, we assume the next Lipschitz conditions:
\begin{description}
	\item[(H4**)] For each $i,j,l=1,\dots,n$, $p=1,\dots,P$, and $q=1,\dots,Q$, there are positive numbers $\gamma^{(1)}_{ijlp}$, $\gamma^{(2)}_{ijlp}$, $\mu^{(1)}_{ijlq}$,  $\mu^{(2)}_{ijlq}$, $\xi_{ijq}$, $\til{\xi}_{ilq}$, $\zeta_{i}$, and $\varsigma_i$ such that
	\begin{eqnarray*}
		|h_{ijlp}(u_1,u_2)-h_{ijlp}(v_1,v_2)|\leq \gamma^{(1)}_{ijlp}|u_1-v_1|+\gamma^{(2)}_{ijlp}|u_2-v_2| \nonumber\\
		|f_{ijlq}(u_1,u_2)-f_{ijlq}(v_1,v_2)|\leq \mu^{(1)}_{ijlq}|u_1-v_1|+\mu^{(2)}_{ijlq}|u_2-v_2| \nonumber\\
	\end{eqnarray*}
	for all $u_1,u_2,v_1,v_2\in\er$, and
	$$
	\begin{array}{ll}
		|g_{ijq}(u)-g_{ijq}(v)|\leq \xi_{ijq}|u-v|,
		&|\til{g}_{ilq}(u)-\til{g}_{ilq}(v)|\leq \til{\xi}_{ilq}|u-v|,\\
		|F_i(u)-F_i(v)|\leq \zeta_i|u-v|,&
		|G_{i}(u)-G_{i}(v)|\leq \varsigma_{i}|u-v|,
	\end{array}
	$$
	for all $u,v\in\er$.
\end{description}

	Now, we state our main stability criterion for model \eqref{3.1}.
\begin{teorema}\label{th 3.1}
   Assume that (H1)-(H3), (H4**), the functions $\tau_{ijp},\til\tau_{ijp}$ are bounded, and there exists $\vartheta>0$ such that
   \begin{eqnarray}\label{4.2}
   	\int_{-\infty}^{0}\Ne^{-\vartheta s}d\eta_{ijq}(s)<+\infty,\Es \int_{-\infty}^{0}\Ne^{-\vartheta s}d\til{\eta}_{ilq}(s)<+\infty.
   \end{eqnarray}
	If there exist $\varepsilon>0$ and $w=(w_1,\ldots,w_n)>0$ such that for all $t\geq0$, and $i=1,\ldots,n$,	
	\begin{eqnarray}\label{4.3}
		\un{a}_i\big(\beta_i(t)+D_i(t)\big)-\ov{a}_i\sum_{j,l=1}^{n}\lefteqn{\left[\dst\sum_{p=1}^{P}\zeta_i|c_{ijlp}(t)|\left(\frac{w_j}{w_i}\gamma^{(1)}_{ijlp}+\frac{w_l}{w_i}\gamma^{(2)}_{ijlp}\right)\right.}\nonumber\\
		&+\dst\left.\sum_{q=1}^{Q}\varsigma_i|d_{ijlq}(t)|\left(\frac{w_j}{w_i}\mu^{(1)}_{ijlq}\xi_{ijq}+\frac{w_l}{w_i}\mu^{(2)}_{ijlq}\til{\xi}_{ilq}\right)\right]>\varepsilon,
	\end{eqnarray}
	then the model \eqref{3.1} is globally exponentially stable.
\end{teorema}
\begin{proof}
	With the change of variables $y_i(t)=w_i^{-1}x_i(t)$, model \eqref{3.1} is transformed into
	 \begin{eqnarray} \label{3.1b}
		y'_i(t)&=&a_i(t,w_iy_i(t))w_i^{-1}\bigg[-b_i(t,w_iy_i(t))+I_i(t)\nonumber\\
		& &+F_i\bigg(\dst\sum_{p=1}^{P}\sum_{j,l=1}^{n}c_{ijlp}(t)h_{ijlp}\big(w_jy_j(t-\tau_{ijp}(t)),w_ly_l(t-{\til{\tau}}_{ilp}(t)\big)\bigg)+G_i\bigg(\dst\sum_{q=1}^{Q}\sum_{j,l=1}^{n}d_{ijlq}(t)\nonumber\\
		& &\cdot f_{ijlq}\left(\int_{-\infty}^0g_{ijq}(w_jy_j(t+s))d\eta_{ijq}(s),\int_{-\infty}^0{\til{g}}_{ilq}(w_ly_l(t+s))d{\til{\eta}}_{ilq}(s)\right)\bigg)\bigg],
	\end{eqnarray}
	for $t\geq 0$, and $i=1,\ldots,n$.
	
    From \eqref{4.3}, there exists $\nu>0$ such that
    \begin{eqnarray}\label{4.4}
    	\un{a}_i\big(\beta_i(t)+D_i(t)\big)-\ov{a}_i\lefteqn{\sum_{j,l=1}^{n}\left[\dst\sum_{p=1}^{P}\zeta_i|c_{ijlp}(t)|\left(\frac{w_j}{w_i}\gamma^{(1)}_{ijlp}+\frac{w_l}{w_i}\gamma^{(2)}_{ijlp}\right)\right.}\nonumber\\
    	&+\dst\left.\sum_{q=1}^{Q}\varsigma_i|d_{ijlq}(t)|\left(\frac{w_j}{w_i}\mu^{(1)}_{ijlq}\xi_{ijq}+\frac{w_l}{w_i}\mu^{(2)}_{ijlq}\til{\xi}_{ilq}\right)\right](1+\nu)>\nu,
    \end{eqnarray}
    for all $t\geq 0$ and $i=1,\dots,n$.

As $\tau_{ijp}$ and $\til{\tau}_{ilp}$ are bounded functions, it is possible to define the non-negative real number
$$
\tau:=\max_{i,j,p}\left(\sup_{t\geq0}\left\{\tau_{ijp}(t),\til\tau_{ijp}(t)\right\}\right).
$$
 As in the proof of \cite[Theorem 4.3]{faria2011general}, from \eqref{4.2}, we can conclude that there exists $\alpha\in(0,\vartheta)$ such that
	\begin{eqnarray}\label{4.6}
		\int_{-\infty}^{0}\Ne^{-\alpha s}d\eta_{ijq}(s)<1+\nu\Es\text{ and }\Es \int_{-\infty}^{0}\Ne^{-\alpha s}d\til{\eta}_{ijq}(s)<1+\nu,
	\end{eqnarray}
	for all $i,j=1,\dots,n$ and $q=1,\dots,Q$.\\
	
	Let $\epsilon:=\min\{\nu,\alpha,\frac{\log(1+\nu)}{\tau+1}\}$ and consider the system \eqref{3.1b} in the phase space ${UC}^n_\epsilon$.
	
		Defining, for each $i=1,\ldots,n$, $\til{a}_i(t,u):=a_i(t,w_iu)$, $\til{b}_i(t,u)=w_i^{-1}b_i(t,w_iu)$, and
	\begin{align*}
		\til{f}_i(t,\phi)&:=w_i^{-1}F_i\bigg(\sum_{p=1}^{P}\sum_{j,l=1}^{n}c_{ijlp}(t)h_{ijlp}\big(w_j\phi_j(-\tau_{ijp}(t)),w_l\phi_l(-\til{\tau}_{ilp}(t))\big)\bigg)+w_i^{-1}I_i(t)\\
		&+w_i^{-1}G_i\Bigg(\sum_{q=1}^{Q}\sum_{j,l=1}^{n}d_{ijlq}(t)f_{ijlq}\bigg(\int_{-\infty}^{0}g_{ijq}(w_j\phi_j(s))d\eta_{ijq}(s),\int_{-\infty}^{0}\til{g}_{ilq}(w_l\phi_l(s))d\til{\eta}_{ilq}(s)\bigg)\Bigg)
	\end{align*}
	for all $u\in\er$ and $\phi=(\phi_1,\ldots,\phi_n)\in UC^n_\epsilon$, model \eqref{3.1b} has the form
	\begin{eqnarray}\label{1.1b}
			y'_i(t)=\til{a}_i(t,y_i(t))\big[-\til{b}_i(t,y_i(t))+\til{f}_i(t,y_t)\big],\Es t\geq 0, \, i=1,\dots,n.
    \end{eqnarray}
For model \eqref{1.1b}, the hypotheses (H1), (H2), and (H3) hold with same constants $\un{a}_i, \ov{a}_i$ and same functions $D_i(t),\beta_i(t)$.

From Theorem \ref{theorem 2.2}, the proof is concluded if hypotheses (H4) and (H5) hold.

%	As shall we see below, $\epsilon\leq \nu$, $\epsilon\leq \theta'$ are necessary to deal with the distributed delay part, and $\epsilon\leq \frac{\log(1+\nu)}{\tau}$  is required to deal with the discrete delay part in the model \eqref{3.1}. We have $e^{\epsilon\tau_{ijp}(t)}\leq e^{\epsilon \tau}\leq e^{\frac{\log(1+\nu)}{\tau}\tau}=e^{\log(1+\nu)}=1+\nu$, similarly $e^{\epsilon\til{\tau}_{ijp}(t)}\leq 1+\nu$.\\
	For $\phi=(\phi_1,\ldots,\phi_n),\psi=(\psi_1,\ldots,\psi_n)\in {UC}^n_\epsilon$, $t\geq0$, and $i=1,\ldots,n$, from (H4**) we have
	\begin{align*}
		|\til{f}_i(t,\phi)&-\til{f}_i(t,\psi)|\leq w_i^{-1}\sum_{j,l=1}^n\left[\zeta_i\sum_{p=1}^{P}|c_{ijlp}(t)|\right.\\
		&\cdot\left|h_{ijlp}\big(w_j\phi_j(-\tau_{ijp}(t)),w_l\phi_l(-\til{\tau}_{ilp}(t))\big)-h_{ijlp}\big(w_j\psi_j(-\tau_{ijp}(t)),w_l\psi_l(-\til{\tau}_{ilp}(t))\big)\right|\\
		&+\varsigma_{i}\sum_{q=1}^{Q}|d_{ijlq}(t)|\left|f_{ijlq}\bigg(\int_{-\infty}^{0}g_{ijq}(w_j\phi_j(s))d\eta_{ijq}(s),\int_{-\infty}^{0}\til{g}_{ilq}(w_l\phi_l(s))d\til{\eta}_{ilq}(s)\bigg)\right.\\
		&-\left.\left.f_{ijlq}\bigg(\int_{-\infty}^{0}g_{ijq}(w_j\psi_j(s))d\eta_{ijq}(s),\int_{-\infty}^{0}\til{g}_{ilq}(w_l\psi_l(s))d\til{\eta}_{ilq}(s)\bigg)\right|\right]\\
		\leq& w_i^{-1}\sum_{j,l=1}^n\left[\zeta_i\sum_{p=1}^{P}|c_{ijlp}(t)|\right.\\
		&\cdot\bigg(\gamma_{ijlp}^{(1)}w_j\big|\phi_j(-\tau_{ijp}(t))-\psi_j(-\tau_{ijp}(t))\big|+\gamma_{ijlp}^{(2)}w_l\big|\phi_l(-\til{\tau}_{ilp}(t))-\psi_l(-\til{\tau}_{ilp}(t))\big|\bigg)\\
		&+\varsigma_{i}\sum_{q=1}^{Q}|d_{ijlq}(t)|\bigg(\mu_{ijlq}^{(1)}\left|\int_{-\infty}^{0}g_{ijq}(w_j\phi_j(s))-g_{ijq}(w_j\psi_j(s))d\eta_{ijq}(s)\right|\\
		&+\mu_{ijlq}^{(2)}\left|\int_{-\infty}^{0}\til{g}_{ilq}(w_l\phi_l(s))-\til{g}_{ilq}(w_l\psi_l(s))d\til{\eta}_{ilq}(s)\right|\bigg)\bigg].
		\end{align*}
	Again from (H4**) and by the monotony of $\eta_{ijq}$ and $\til{\eta}_{ijq}$ we obtain,
		\begin{align}\label{eq:prova-parcial}
		|\til{f}_i(t,\phi)-\til{f}_i(t,\psi)|\leq& \sum_{j,l=1}^n\left[\zeta_i\sum_{p=1}^{P}|c_{ijlp}(t)|\bigg(\gamma_{ijlp}^{(1)}\frac{w_j}{w_i}\big|\phi_j(-\tau_{ijp}(t))-\psi_j(-\tau_{ijp}(t))\big|\right.\nonumber\\
		 &+\gamma_{ijlp}^{(2)}\frac{w_l}{w_i}\big|\phi_l(-\til{\tau}_{ilp}(t))-\psi_l(-\til{\tau}_{ilp}(t))\big|\bigg)\nonumber\\
		 &+\varsigma_{i}\sum_{q=1}^{Q}|d_{ijlq}(t)|\bigg(\mu_{ijlq}^{(1)}\int_{-\infty}^{0}\xi_{ijq}\frac{w_j}{w_i}|\phi_j(s)-\psi_j(s)|d\eta_{ijq}(s)\nonumber\\
		 &+\left.\mu_{ijlq}^{(2)}\int_{-\infty}^{0}\til{\xi}_{ilq}\frac{w_l}{w_i}|\phi_l(s)-\psi_l(s)|d\til{\eta}_{ilq}(s)\bigg)\right]
	\end{align}
	and consequently
	\begin{align*}
		|\til{f}_i(t,\phi)-\til{f}_i(t,\psi)|\leq&
		\sum_{j,l=1}^n\left[\zeta_i\sum_{p=1}^{P}|c_{ijlp}(t)|\bigg(\gamma_{ijlp}^{(1)}\frac{w_j}{w_i}\frac{\big|(\phi_j-\psi_j)(-\tau_{ijp}(t))\big|}{\Ne^{-\epsilon\left(-\tau_{ijp}(t)\right)}}\Ne^{\epsilon\tau_{ijp}(t)}\right.\\
		&+\left.\gamma_{ijlp}^{(2)}\frac{w_l}{w_i}\frac{\big|(\phi_l-\psi_l)(-\til{\tau}_{ilp}(t))\big|}{\Ne^{-\epsilon\left(-\til{\tau}_{ilp}(t)\right)}}\Ne^{\epsilon\til{\tau}_{ilp}(t)}\right)\\
		&+\varsigma_{i}\sum_{q=1}^{Q}|d_{ijlq}(t)|\bigg(\mu_{ijlq}^{(1)}\int_{-\infty}^{0}\xi_{ijq}\frac{w_j}{w_i}\frac{|(\phi_j-\psi_j)(s)|}{\Ne^{-\epsilon s}}\Ne^{-\epsilon s}d\eta_{ijq}(s)\\
		&+\left.\left.\mu_{ijlq}^{(2)}\int_{-\infty}^{0}\til{\xi}_{ilq}\frac{w_l}{w_i}\frac{|(\phi_l-\psi_l)(s)|}{\Ne^{-\epsilon s}}\Ne^{-\epsilon s}d\til{\eta}_{ilq}(s)\right)\right]\\
		\leq&\sum_{j,l=1}^n\left[\zeta_i\sum_{p=1}^{P}|c_{ijlp}(t)|\left(\gamma_{ijlp}^{(1)}\frac{w_j}{w_i}\|\phi-\psi\|_\epsilon\Ne^{\epsilon\tau_{ijp}(t)}+\gamma_{ijlp}^{(2)}\frac{w_l}{w_i}\|\phi-\psi\|_\epsilon\Ne^{\epsilon\til{\tau}_{ilp}(t)}\right)\right.\\
		&+\varsigma_{i}\sum_{q=1}^{Q}|d_{ijlq}(t)|\bigg(\mu_{ijlq}^{(1)}\int_{-\infty}^{0}\xi_{ijq}\frac{w_j}{w_i}\|\phi-\psi\|_\epsilon\Ne^{-\epsilon s}d\eta_{ijq}(s)\\
		&+\left.\mu_{ijlq}^{(2)}\int_{-\infty}^{0}\til{\xi}_{ilq}\frac{w_l}{w_i}\|\phi-\psi\|_\epsilon \Ne^{-\epsilon s}d\til{\eta}_{ilq}(s)\bigg)\right]\\
		\leq&\|\phi-\psi\|_\epsilon\sum_{j,l=1}^n\left[\zeta_i\sum_{p=1}^{P}|c_{ijlp}(t)|\left(\gamma_{ijlp}^{(1)}\frac{w_j}{w_i}+\gamma_{ijlp}^{(2)}\frac{w_l}{w_i}\right)\Ne^{\epsilon\tau}\right.\\
		&+\left.\varsigma_{i}\sum_{q=1}^{Q}|d_{ijlq}(t)|\bigg(\mu_{ijlq}^{(1)}\xi_{ijq}\frac{w_j}{w_i}\int_{-\infty}^{0}\Ne^{-\epsilon s}d\eta_{ijq}(s)+\mu_{ijlq}^{(2)}\til{\xi}_{ilq}\frac{w_l}{w_i}\int_{-\infty}^{0} \Ne^{-\epsilon s}d\til{\eta}_{ilq}(s)\bigg)\right].
	\end{align*}
     As $\epsilon\leq\alpha$ from \eqref{4.6} we have
     $$
     \int_{-\infty}^{0}\Ne^{-\epsilon s}d\eta_{ijq}(s)<1+\nu\Es\text{ and }\Es\int_{-\infty}^{0}\Ne^{-\epsilon s}d\til{\eta}_{ijq}(s)<1+\nu,
     $$
     for all $i,j=1,\ldots,n$. As $\epsilon\leq\dst\frac{\log(1+\nu)}{\tau+1}$, then we also have
	$$
	  \Ne^{\epsilon\tau}<\Ne^{\epsilon(\tau+1)}\leq1+\nu.
	$$
	Consequently
	\begin{align*}
		|\til{f}_i(t,\phi)-\til{f}_i(t,\psi)|\leq&\left[\sum_{j,l=1}^n\left(\zeta_i\sum_{p=1}^{P}|c_{ijlp}(t)|\left(\gamma_{ijlp}^{(1)}\frac{w_j}{w_i}+\gamma_{ijlp}^{(2)}\frac{w_l}{w_i}\right)\right.\right.\\
		&+\left.\left.\varsigma_{i}\sum_{q=1}^{Q}|d_{ijlq}(t)|\bigg(\mu_{ijlq}^{(1)}\xi_{ijq}\frac{w_j}{w_i}+\mu_{ijlq}^{(2)}\til{\xi}_{ilq}\frac{w_l}{w_i}\bigg)\right)(1+\nu)\right]\|\phi-\psi\|_\epsilon,
	\end{align*}
 and hypothesis (H4) holds with
 $$
   \mathcal{L}_i(t)=\sum_{j,l=1}^n\left(\zeta_i\sum_{p=1}^{P}|c_{ijlp}(t)|\left(\gamma_{ijlp}^{(1)}\frac{w_j}{w_i}+\gamma_{ijlp}^{(2)}\frac{w_l}{w_i}\right)
   +\varsigma_{i}\sum_{q=1}^{Q}|d_{ijlq}(t)|\bigg(\mu_{ijlq}^{(1)}\xi_{ijq}\frac{w_j}{w_i}+\mu_{ijlq}^{(2)}\til{\xi}_{ilq}\frac{w_l}{w_i}\bigg)\right)(1+\nu)
 $$
 for all $i=1,\ldots,n$.

 As $\epsilon\leq\nu$ and from \eqref{4.4}, the hypothesis (H5) also holds and the proof is concluded.
\end{proof}

Now, we assume that the model \eqref{3.1} is periodic, i.e. there is $\omega>0$ such that
\begin{description}
	\item[{(H2**)}] There is $\omega>0$ such that, for each $i,j,l=1,\ldots,n$, $p=1,\ldots,P$, and \mbox{$q=1,\ldots,Q$},
	{$$
	\begin{array}{llll}
	a_i(t,u)=a_i(t+\omega,u), & c_{ijlp}(t)=c_{ijlp}(t+\omega), & \tau_{ijp}(t)=\tau_{ijp}(t+\omega),\\
	b_i(t,u)=b_i(t+\omega,u), & d_{ijlq}(t)=d_{ijlq}(t+\omega), & \til{\tau}_{ijp}(t)=\til{\tau}_{ijp}(t+\omega), \text{ and }\\
    I_i(t)=I_i(t+\omega) & & &
	\end{array}
	$$}
	for all $t\geq0$ and $u\in\er$.
\end{description}

\begin{teorema}\label{teo:existencia-sol-perio-model-part}
	Assume the hypotheses {(H2**)}, (H1), (H3*), and (H4**).
	
	If there exists $w=(w_1,\ldots,w_n)>0$ such that for all $t\in[0,\omega]$, and $i=1,\ldots,n$,	
	\begin{eqnarray}\label{4.3b}
		\beta_i(t)>\sum_{j,l=1}^{n}\lefteqn{\left[\dst\sum_{p=1}^{P}\zeta_i|c_{ijlp}(t)|\left(\frac{w_j}{w_i}\gamma^{(1)}_{ijlp}+\frac{w_l}{w_i}\gamma^{(2)}_{ijlp}\right)\right.}\nonumber\\
		&+\dst\left.\sum_{q=1}^{Q}\varsigma_i|d_{ijlq}(t)|\left(\frac{w_j}{w_i}\mu^{(1)}_{ijlq}\xi_{ijq}+\frac{w_l}{w_i}\mu^{(2)}_{ijlq}\til{\xi}_{ilq}\right)\right],
	\end{eqnarray}
	then the model \eqref{3.1} has an $\omega-$periodic solution.	
\end{teorema}
\begin{proof}
 As in previous proof, the change of variables $y_i(t)=w_i^{-1}x_i(t)$ transforms model \eqref{3.1} into \eqref{3.1b}.

 Considering model \eqref{3.1b} in the phase space $UC_\epsilon^n$, for some $\epsilon>0$, it has the form \eqref{1.1b}.

 Proceeding as in the previous proof, functions $\til{f}_i$ verify \eqref{eq:prova-parcial} for all $\phi=(\phi_1,\ldots,\phi_n),\psi=(\psi_1,\ldots,\psi_n)\in BC$ and $t\geq0$. Consequently
 \begin{align*}
 	|\til{f}_i(t,\phi)-\til{f}_i(t,\psi)|\leq&\|\phi-\psi\| \sum_{j,l=1}^n\left[\zeta_i\sum_{p=1}^{P}|c_{ijlp}(t)|\bigg(\gamma_{ijlp}^{(1)}\frac{w_j}{w_i}+\gamma_{ijlp}^{(2)}\frac{w_l}{w_i}\bigg)\right.\\
 	&+\left.\varsigma_{i}\sum_{q=1}^{Q}|d_{ijlq}(t)|\bigg(\mu_{ijlq}^{(1)}\xi_{ijq}\frac{w_j}{w_i}\int_{-\infty}^{0}d\eta_{ijq}(s)+\mu_{ijlq}^{(2)}\til{\xi}_{ilq}\frac{w_l}{w_i}\int_{-\infty}^{0}d\til{\eta}_{ilq}(s)\bigg)\right],\\
 \end{align*}
 and from the properties of $\eta_{ijq}$ and $\til{\eta}_{ijq}$ we obtain
 $$
   |\til{f}_i(t,\phi)-\til{f}_i(t,\psi)|\leq\mathcal{L}_i(t)\|\phi-\psi\|,
 $$
 with
 $$
   \mathcal{L}_i(t)=\sum_{j,l=1}^n\left[\zeta_i\sum_{p=1}^{P}|c_{ijlp}(t)|\bigg(\gamma_{ijlp}^{(1)}\frac{w_j}{w_i}+\gamma_{ijlp}^{(2)}\frac{w_l}{w_i}\bigg)+\varsigma_{i}\sum_{q=1}^{Q}|d_{ijlq}(t)|\bigg(\mu_{ijlq}^{(1)}\xi_{ijq}\frac{w_j}{w_i}+\mu_{ijlq}^{(2)}\til{\xi}_{ilq}\frac{w_l}{w_i}\bigg)\right],
 $$
 thus (H4*) holds for model \eqref{1.1b}.

 By hypothesis \eqref{4.3b}, (H5*) also holds and the conclusion follows from Theorem \ref{teo:existencia-orbita-periodica}.
\end{proof}

Immediately from Theorems  \ref{th 3.1} and \ref{teo:existencia-sol-perio-model-part}, we have the following result.
\begin{corolario}\label{cor:existencia+estabilidade-modelo-CG}
	Assume {(H1)}, (H2) with $D_i$ an $\omega-$periodic continuous function, {(H2**)}, (H3*), (H4**), and \eqref{4.2}.
	
	If there exists $w=(w_1,\ldots,w_n)>0$ such that, for all $t\in[0,\omega]$ and $i=1,\ldots,n$ inequality \eqref{4.3b} holds and
	\begin{eqnarray}\label{4.3c}
		\un{a}_i\big(\beta_i(t)+D_i(t)\big)>\ov{a}_i\sum_{j,l=1}^{n}\lefteqn{\left[\dst\sum_{p=1}^{P}\zeta_i|c_{ijlp}(t)|\left(\frac{w_j}{w_i}\gamma^{(1)}_{ijlp}+\frac{w_l}{w_i}\gamma^{(2)}_{ijlp}\right)\right.}\nonumber\\
		&+\dst\left.\sum_{q=1}^{Q}\varsigma_i|d_{ijlq}(t)|\left(\frac{w_j}{w_i}\mu^{(1)}_{ijlq}\xi_{ijq}+\frac{w_l}{w_i}\mu^{(2)}_{ijlq}\til{\xi}_{ilq}\right)\right],
	\end{eqnarray}
then the model \eqref{3.1} has an $\omega-$periodic solution which is globally exponentially stable.
\end{corolario}
\begin{proof}
	From {(H2**)} functions $\tau_{ijp}$, $\til{\tau}_{ijp}$ are bounded.
	
	Moreover, from {(H2**)} and (H3*) we know that $\beta_i$, $c_{ijlp}$, and $d_{ijlq}$ are $\omega-$periodic functions. As $D_i$ are also $\omega-$periodic, then there is $\varepsilon>0$ such that inequality \eqref{4.3} holds and the conclusion comes from  Theorems  \ref{th 3.1} and \ref{teo:existencia-sol-perio-model-part}.
\end{proof}

Now, we consider model \eqref{3.1} with amplifications functions, $a_i$, do not explicitly depend on time $t$, i.e.
 \begin{eqnarray} \label{3.1c}
	x'_i(t)&=&a_i(x_i(t))\bigg[-b_i(t,x_i(t))+F_i\bigg(\dst\sum_{p=1}^{P}\sum_{j,l=1}^{n}c_{ijlp}(t)h_{ijlp}\big(x_j(t-\tau_{ijp}(t)),x_l(t-{\til{\tau}}_{ilp}(t)\big)\bigg)\nonumber\\
	& &+G_i\bigg(\dst\sum_{q=1}^{Q}\sum_{j,l=1}^{n}d_{ijlq}(t)f_{ijlq}\left(\int_{-\infty}^0g_{ijq}(x_j(t+s))d\eta_{ijq}(s),\int_{-\infty}^0{\til{g}}_{ilq}(x_l(t+s))d{\til{\eta}}_{ilq}(s)\right)\bigg)\nonumber\\
	& &+I_i(t)\bigg],\Es t\geq 0,\Es \,i=1,\ldots,n,
\end{eqnarray}

From Corollary \ref{cor:existencia+estabilidade-modelo-CG} we have the following result.
\begin{corolario}\label{cor:existencia+estabilidade-modelo-ampl-auto}
		Assume (H1), {(H2**)}, (H3*), (H4**), and \eqref{4.2}.
	
	If there exists $w=(w_1,\ldots,w_n)>0$ such that, for all $t\in[0,\omega]$ and $i=1,\ldots,n$,
	\begin{eqnarray}\label{4.3d}
		\un{a}_i\beta_i(t)>\ov{a}_i\sum_{j,l=1}^{n}\lefteqn{\left[\dst\sum_{p=1}^{P}\zeta_i|c_{ijlp}(t)|\left(\frac{w_j}{w_i}\gamma^{(1)}_{ijlp}+\frac{w_l}{w_i}\gamma^{(2)}_{ijlp}\right)\right.}\nonumber\\
		&+\dst\left.\sum_{q=1}^{Q}\varsigma_i|d_{ijlq}(t)|\left(\frac{w_j}{w_i}\mu^{(1)}_{ijlq}\xi_{ijq}+\frac{w_l}{w_i}\mu^{(2)}_{ijlq}\til{\xi}_{ilq}\right)\right],
	\end{eqnarray}
	then the model \eqref{3.1c} has an $\omega-$periodic solution which is globally exponentially stable.
\end{corolario}
\begin{proof}
	Noting that $a_i(t,u)=a_i(u)$ for all $t,u\in\er$ and $i=1,\ldots,n$, the hypothesis (H2) trivially holds with $D_i(t)=0$. Consequently inequality \eqref{4.3d} implies \eqref{4.3b} and the result comes from Corollary \ref{cor:existencia+estabilidade-modelo-CG}.
\end{proof}

For model \eqref{3.1c} under the hypotheses {(H2**)}, (H1), (H3*), and (H4**), consider the constants
\begin{eqnarray}\label{eq:definicao-constantes}
  \un{\beta}_i:=\min_{t\in[0,\omega]}\beta_i(t),\Es\ov{c}_{ijlp}:=\max_{t\in[0,\omega]}c_{ijlp}(t),\Es\text{and}\Es \ov{d}_{ijlq}:=\max_{t\in[0,\omega]}d_{ijlq}(t),
\end{eqnarray}
for each $i,j=1,\ldots,n$, $p=1,\ldots,P$, $q=1,\ldots,Q$, and the
square real matrix $\mathcal{M}$ defined by
\begin{eqnarray}\label{def-M-matrix}
	\mathcal{M}:=diag\big(\un{a}_1\un{\beta}_1,\ldots,\un{a}_n\un{\beta}_n\big)-\big[\mathfrak{m}_{ij}\big]_{i,j=1}^n,
\end{eqnarray}
where, for each $i,j=1,\ldots,n$,
$$
\mathfrak{m}_{ij}:=\ov{a}_i\sum_{l=1}^n\left(\zeta_i\sum_{p=1}^P\left(\ov{c}_{ijlp}\gamma_{ijlp}^{(1)}+\ov{c}_{iljp}\gamma_{iljp}^{(2)}\right)+\varsigma_i\sum_{q=1}^Q\left(\ov{d}_{ijlq}\mu_{ijlq}^{(1)}\xi_{ijq}+\ov{d}_{iljq}\mu_{iljq}^{(2)}\til{\xi}_{ijq}\right)\right).
$$

\begin{corolario}\label{cor:crit-M-matrix}
	Assume (H1), {(H2**)}, (H3*), (H4**), and \eqref{4.2}.
	
	If $\mathcal{M}$ is a non-singular M-matrix, then the model \eqref{3.1c} has an $\omega-$periodic solution which is globally exponentially stable.
\end{corolario}
\begin{proof}
	As $\mathcal{M}$ is a non-singular M-matrix, then (see \cite{fiedler2008special}) there exists $w=(w_1,\ldots,w_n)>0$ such that $\mathcal{M}w^T>0$, i.e.,
	$$
	  \un{a}_i\un{\beta}_iw_i>\sum_{j=1}^nw_j\left[\ov{a}_i\sum_{l=1}^n\left(\zeta_i\sum_{p=1}^P\left(\ov{c}_{ijlp}\gamma_{ijlp}^{(1)}+\ov{c}_{iljp}\gamma_{iljp}^{(2)}\right)+\varsigma_i\sum_{q=1}^Q\left(\ov{d}_{ijlq}\mu_{ijlq}^{(1)}\xi_{ijq}+\ov{d}_{iljq}\mu_{iljq}^{(2)}\til{\xi}_{ijq}\right)\right)\right],
	$$
	for all $i=1,\ldots,n$, which is equivalent to
	\begin{eqnarray*}
		\un{a}_i\un{\beta}_i>\ov{a}_i\sum_{j,l=1}^{n}\lefteqn{\left[\zeta_i\dst\sum_{p=1}^{P}\left(\ov{c}_{ijlp}\frac{w_j}{w_i}\gamma^{(1)}_{ijlp}+\ov{c}_{iljp}\gamma^{(2)}_{iljp}\frac{w_j}{w_i}\right)\right.}\nonumber\\
		&+\dst\left.\varsigma_i\sum_{q=1}^{Q}\left(\ov{d}_{ijlq}\frac{w_j}{w_i}\mu^{(1)}_{ijlq}\xi_{ijq}+\ov{d}_{iljq}\mu^{(2)}_{iljq}\til{\xi}_{ijq}\frac{w_j}{w_i}\right)\right]
	\end{eqnarray*}
and consequently
  \begin{eqnarray}\label{proof:M-matrix-model-amp-auto}
  	\un{a}_i\un{\beta}_i>\ov{a}_i\sum_{j,l=1}^{n}\lefteqn{\left[\zeta_i\dst\sum_{p=1}^{P}\left(\ov{c}_{ijlp}\frac{w_j}{w_i}\gamma^{(1)}_{ijlp}+\ov{c}_{ijlp}\gamma^{(2)}_{ijlp}\frac{w_l}{w_i}\right)\right.}\nonumber\\
  	&+\dst\left.\varsigma_i\sum_{q=1}^{Q}\left(\ov{d}_{ijlq}\frac{w_j}{w_i}\mu^{(1)}_{ijlq}\xi_{ijq}+\ov{d}_{ijlq}\mu^{(2)}_{ijlq}\til{\xi}_{ilq}\frac{w_l}{w_i}\right)\right].
  \end{eqnarray}
From \eqref{eq:definicao-constantes} and \eqref{proof:M-matrix-model-amp-auto} we obtain \eqref{4.3d}. Now the result follows from Corollary \ref{cor:existencia+estabilidade-modelo-ampl-auto}.
\end{proof}

\exemplo Consider the following low-order Cohen-Grossberg neural network model
\begin{eqnarray}\label{eq:static-model}
	x'(t)=a_i(x_i(t))\left[-b_i(t,x_i(t))+G_i\left(\sum_{j=1}^nc_{ij}(t)\int_0^{+\infty}x_j(t-u)K_{ij}(u)du\right)\right],
\end{eqnarray}
for $t\geq0$ and $i=1,\ldots,n$, where $a_i:\er\to(0,+\infty)$, $b_i:[0,+\infty)\times\er\to\er$, $c_{ij}:[0,+\infty)\to\er$, $G_i:\er\to\er$, and $K_{ij}:[0,+\infty)\to[0,+\infty)$ are continuous functions such that
\begin{eqnarray}\label{ex:kernel1}
  \int_0^{+\infty}K_{ij}(u)du=1,
\end{eqnarray}
for all $i,j=1,\ldots,n$.

The model \eqref{eq:static-model} is a generalization of the following autonomous static neural network model
 \begin{eqnarray}\label{eq:static-model-INcube}
 	x'(t)=-x_i(t)+G_i\left(\sum_{j=1}^nc_{ij}\int_0^{+\infty}x_j(t-u)K_{ij}(u)du\right),\,t\geq0,\,i=1,\ldots,n,
 \end{eqnarray}
whose the existence and global asymptotic stability of an equilibrium point was studied in \cite{ncube2020existence}.

Defining, for each $i,j=1,\ldots,n$, $\eta_{ij}:(-\infty,0]\to\er$ by
\begin{eqnarray}\label{eq:-defin-eta}
  \eta_{ij}(s)=\int_{-\infty}^sK_{ij}(-v)dv,\Es s\in(-\infty,0]
\end{eqnarray}
we have $\eta_{ij}$ non-decreasing and, from \eqref{ex:kernel1}, $\eta_{ij}(0)-\eta_{ij}(-\infty)=1$. Consequently, the model \eqref{eq:static-model} can be written in the form
 \begin{eqnarray*}\label{eq:static-model2}
 	x'(t)=a_i(x_i(t))\left[-b_i(t,x_i(t))+G_i\left(\sum_{j=1}^nc_{ij}(t)\int_{-\infty}^0x_j(t+s)d\eta_{ij}(s)\right)\right],\,t\geq0,\,i=1,\ldots,n,
 \end{eqnarray*}
which is a particular situation of \eqref{3.1c}. Consequently, from Corollary \ref{cor:existencia+estabilidade-modelo-ampl-auto}, we obtain the following result.
\begin{corolario}
	Assume (H1), (H3*) and, for each $i,j=1,\ldots,n$, the functions {\mbox{$t \mapsto b_i(t,u)$}} and $c_{ij}$ is $\omega-$periodic,  the function $G_i$ is Lipschitz with Lipschitz constant $\varsigma_i>0$, and there is $\alpha>0$ such that
	\begin{eqnarray}\label{eq:restricao-kernel}
	  \int_0^{+\infty}K_{ij}(u)\Ne^{\alpha u}du<+\infty.
	\end{eqnarray}
	If there exists $w=(w_1,\ldots,w_n)>0$ such that, for all $t\in[0,\omega]$ and $i=1,\ldots,n$,
	\begin{eqnarray}\label{4.3d-static}
		\un{a}_i\beta_i(t)w_i>\ov{a}_i\sum_{j=1}^{n}\varsigma_i|c_{ij}(t)|w_j,\Es t\geq0,\,i=1,\ldots,n,
	\end{eqnarray}
	then the model \eqref{eq:static-model} has an $\omega-$periodic solution which is globally exponentially stable.
\end{corolario}
For model \eqref{eq:static-model-INcube}, condition \eqref{4.3d-static} reads as
 $$
 	w_i>\sum_{j=1}^{n}\varsigma_i|c_{ij}|w_j,\Es t\geq0,\,i=1,\ldots,n,
 $$
 which is equivalent to matrix
 \begin{eqnarray*}\label{eq:M-matrix-N}
   \mathcal{N}=I_n-\big[\varsigma_i|c_{ij}|\big]_{i,j=1}^n,
\end{eqnarray*}
where $I_n$ denotes the identity matrix of $n$-dimension, being a non-singular M-matrix (see [Fidler]). Consequently, we also have the following result.
 \begin{corolario}
 	For each $i,j=1,\ldots,n$ assume that the function $G_i$ is Lipschitz with Lipschitz constant $\varsigma_i>0$ and \eqref{eq:restricao-kernel}.
 	
 	If $\mathcal{N}$ is a non-singular M-matrix,
 	then the model \eqref{eq:static-model-INcube} has an equilibrium point which is globally exponentially stable.
 \end{corolario}

\begin{rem}
	We remark that in \cite{ncube2020existence} the existence and global asymptotic stability of an equilibrium point of \eqref{eq:static-model-INcube} was obtained assuming stronger conditions over $G_i$ than being Lipschitz, $\mathcal{N}$ be a non-singular M-matrix, and
	$$
	  \int_0^{+\infty}uK_{ij}(u)du<+\infty.
	$$
	instead of \eqref{eq:restricao-kernel}.	
\end{rem}

\exemplo Consider the following low-order Cohen-Grossberg neural network model,
\begin{eqnarray}\label{eq:L-O-Cohen-Grossberg-model}
	x'(t)\lefteqn{=a_i(t,x_i(t))\bigg[-b_i(t,x_i(t))+\sum_{j=1}^nc_{ij1}(t)h_{ij1}(x_j(t))+\sum_{j=1}^nc_{ij2}(t)h_{ij2}(x_j(t-\tau_{ij}(t)))}\nonumber\\
	&+\left.\dst\sum_{j=1}^nd_{ij}(t)\int_0^{+\infty}g_{ij}(x_j(t-u))K_{ij}(u)du+I_i(t)\right],\,t\geq0,\,i=1,\ldots,n,
\end{eqnarray}
where, for each $i,j=1,\ldots,n$, $a_i:[0,+\infty)\times\er\to(0,+\infty)$, $b_i:[0,+\infty)\times\er\to\er$, $c_{ij1},c_{ij2},d_{ij},I_i:[0,+\infty)\to\er$, $\tau_{ij}:[0,+\infty)\to[0,+\infty)$, $h_{ij1},h_{ij2},g_{ij}:\er\to\er$, and $K_{ij}:[0,+\infty)\to[0,+\infty)$ are continuous functions such that $K_{ij}$ verifies \eqref{ex:kernel1}.

Sufficient conditions for the exponential stability of \eqref{eq:L-O-Cohen-Grossberg-model} were obtained in \cite{oliveira2017global,liu2008new}.

The existence and global asymptotic stability of a periodic solution of  \eqref{eq:L-O-Cohen-Grossberg-model}, with finite delays, were studied in \cite{hien2014existence}.

The following particular situation of \eqref{eq:L-O-Cohen-Grossberg-model}
\begin{eqnarray}\label{eq:L-O-Cohen-Grossberg-model2}
	x'(t)\lefteqn{=a_i(x_i(t))\bigg[-b_i(x_i(t))+\sum_{j=1}^nc_{ij1}(t)h_{j1}(x_j(t))+\sum_{j=1}^nc_{ij2}(t)h_{j2}(x_j(t-\tau_{ij}(t)))}\nonumber\\
	&+\left.\dst\sum_{j=1}^nd_{ij}(t)\int_0^{+\infty}g_{j}(x_j(t-u))K_{ij}(u)du+I_i(t)\right],\,t\geq0,\,i=1,\ldots,n,
\end{eqnarray}
was studied in \cite{zhu2016existence}, where conditions for the existence and global exponential stability of a pseudo almost automorphic solution were established.

Considering  the definition of the bounded variation function $\eta_{ij}$ as in \eqref{eq:-defin-eta}, model \eqref{eq:L-O-Cohen-Grossberg-model} is a particular situation of \eqref{3.1}, thus from Theorem \ref{th 3.1} we obtain the following stability criterion
\begin{corolario}\label{cor-estabilidade-LO-CG}
   Assume that (H1)-(H3), and, for each $i,j=1,\ldots,n$, $\tau_{ij}$ is bounded, $K_{ij}$ verifies \eqref{ex:kernel1} and {\eqref{eq:restricao-kernel}}, and $h_{ij1},h_{ij2},g_{ij}$ are Lipschitz functions with Lipschitz constants $\gamma_{ij1},\gamma_{ij2},\xi_{ij}>0$ respectively.

  If there exist $\varepsilon>0$ and $w=(w_1,\ldots,w_n)>0$ such that, for all $t\geq0$ and $i=1,\ldots,n$,	
  \begin{eqnarray}\label{4.3e}
  	\un{a}_i\big(\beta_i(t)+D_i(t)\big)w_i-\ov{a}_i\sum_{j=1}^{n}\bigg[\big(\dst|c_{ij1}(t)|\gamma_{ij1}+|c_{ij2}(t)|\gamma_{ij2}+|d_{ij}(t)|\xi_{ij}\big)w_j\bigg]>\varepsilon,
  \end{eqnarray}
  then the model \eqref{eq:L-O-Cohen-Grossberg-model} is globally exponentially stable.
\end{corolario}

\begin{rem}\label{rem-jj2} In \cite{oliveira2017global} the exponential stability of \eqref{eq:L-O-Cohen-Grossberg-model} was established assuming hypotheses in  Corollary \ref{cor-estabilidade-LO-CG} with the condition
	\begin{eqnarray}\label{4.3e2}
		\un{a}_i\big(\beta_i(t)+D_i(t)\big)-\sum_{j=1}^{n}\ov{a}_j\bigg[\big(\dst|c_{ij1}(t)|\gamma_{ij1}+|c_{ij2}(t)|\gamma_{ij2}+|d_{ij}(t)|\xi_{ij}\big)\frac{w_j}{w_i}\bigg]>\varepsilon,
	\end{eqnarray}
instead of \eqref{4.3e}.
\end{rem}
Model \eqref{eq:L-O-Cohen-Grossberg-model2} is a particular situation of \eqref{3.1c}, from Corollaries \ref{cor:existencia+estabilidade-modelo-ampl-auto} and \ref{cor-estabilidade-LO-CG} and Remark \ref{rem-jj2} we obtain the following result
\begin{corolario}\label{cor-comparar-comp-Zhu2016}
 Assume (H1), (H3*) with $\beta_i(t)\equiv\beta_i$, and, for each $i,j=1,\ldots,n$, $c_{ij1},c_{ij2},\tau_{ij},d_{ij}, I_i$ are $\omega-$periodic for some $\omega>0$, $K_{ij}$ verifies \eqref{ex:kernel1} and \eqref{eq:restricao-kernel}, and $h_{j1},h_{j2},g_{j}$ are Lipschitz functions with Lipschitz constants $\gamma_{j1},\gamma_{j2},\xi_{j}>0$ respectively.

If there exists $w=(w_1,\ldots,w_n)>0$ such that one of the following conditions
\begin{eqnarray}\label{4.3e3}
	\un{a}_i\beta_i w_i-\ov{a}_i\sum_{j=1}^{n}\bigg[\big(\dst|c_{ij1}(t)|\gamma_{j1}+|c_{ij2}(t)|\gamma_{j2}+|d_{ij}(t)|\xi_{j}\big)w_j\bigg]>0,
\end{eqnarray}
or
\begin{eqnarray}\label{4.3e4}
	\un{a}_i\beta_i w_i-\sum_{j=1}^{n}\ov{a}_j\bigg[\big(\dst|c_{ij1}(t)|\gamma_{j1}+|c_{ij2}(t)|\gamma_{j2}+|d_{ij}(t)|\xi_{j}\big)w_j\bigg]>0,
\end{eqnarray}
holds for all $t\in[0,\omega]$ and $i=1,\ldots,n$,
then the model \eqref{eq:L-O-Cohen-Grossberg-model2} has an $\omega-$periodic solution which is globally exponentially stable.
\end{corolario}

\begin{rem}
  In \cite{zhu2016existence} the existence of a unique pseudo almost automorphic solution
  of model \eqref{eq:L-O-Cohen-Grossberg-model2} with $c_{ij1},c_{ij2},\tau_{ij},d_{ij}$ being
   pseudo almost automorphic functions was obtained assuming (H1), (H3*) with $\beta_i(t)\equiv\beta_i$, $K_{ij}$ verifying  \eqref{ex:kernel1}, \eqref{eq:restricao-kernel}, and $h_{j1},h_{j2},g_{j}$ are Lipschitz functions with Lipschitz constants $\gamma_{j1},\gamma_{j2},\xi_{j}>0$, and
   \begin{eqnarray}\label{4.3e5}
   	\un{a}_i\beta_i w_i-\sum_{j=1}^{n}\ov{a}_j\bigg[\big(\dst\ov{c}_{ij1}\gamma_{j1}+\ov{c}_{ij2}\gamma_{j2}+\ov{d}_{ij}\xi_{j}\big)w_j\bigg]>0,
   \end{eqnarray}
where $\ov{c}_{ijp}=\sup|c_{ij1}(t)|$, $\ov{d}_{ij}=\sup|d_{ij}(t)|$ for $i,j=1,\ldots,n$ and $p=1,2$.

All periodic functions are pseudo almost automorphic functions. Thus Corollary \ref{cor-comparar-comp-Zhu2016} is not a generalization of \cite[Theorem 3.1]{zhu2016existence}. However, in case of \eqref{eq:L-O-Cohen-Grossberg-model2} being a periodic model, the {existence criterium in Corollary \ref{cor-comparar-comp-Zhu2016} is} better than {the corresponding criterium} in  \cite[Theorem 3.1]{zhu2016existence}.
\end{rem}

\exemplo Consider the following high-order Cohen-Grossberg neural network model,
\begin{eqnarray}\label{eq:H-O-Cohen-Grossberg-model}
	x'(t)\lefteqn{=a_i(x_i(t))\bigg[-b_i(x_i(t))+\sum_{j=1}^nc_{ij}(t)f_j(\rho_jx_j(t))+\sum_{j=1}^nd_{ij11}(t)f_{j}\left(\rho_j\int_0^{+\infty}K_{ij}(u)x_j(t-u)du\right)}\nonumber\\
	&+\left.\dst\sum_{j,l=1}^nd_{ijl2}(t)f_j\left(\rho_j\int_0^{+\infty}K_{ij}(u)x_j(t-u)du\right)f_l\left(\rho_l\int_0^{+\infty}K_{il}(u)x_l(t-u)du\right)+I_i(t)\right],\nonumber\\
\end{eqnarray}
for $\,t\geq0,\,i=1,\ldots,n,$ where, for each $i,j,l=1,\ldots,n$ and $q=1,2$, $\rho_i>0$, $a_i:\er\to(0,+\infty)$, $b_i:\er\to\er$, $c_{ij},d_{ijlq},I_i:[0,+\infty)\to\er$, $f_{j}:\er\to\er$, and $K_{ij}:[0,+\infty)\to[0,+\infty)$ are continuous functions such that $K_{ij}$ verifies \eqref{ex:kernel1}.

The existence and global exponential stability of a periodic solution of \eqref{eq:H-O-Cohen-Grossberg-model} were studied in \cite{liu2011periodic}.

Considering  the definition of the bounded variation functions $\eta_{ij}$ as in \eqref{eq:-defin-eta}, model \eqref{eq:H-O-Cohen-Grossberg-model} is a particular situation of \eqref{3.1c}, thus from Corollary \ref{cor:existencia+estabilidade-modelo-ampl-auto} we obtain the following stability criterion.
\begin{corolario}\label{cor-estabilidade-HO-CG}
	 Assume (H1), (H3*) with $\beta_i(t)\equiv\beta_i$, and, for each $i,j,l=1,\ldots,n$ and $q=1,2$, $c_{ij},d_{ijlq}, I_{i}$ are $\omega-$periodic for some $\omega>0$, $K_{ij}$ verifies \eqref{ex:kernel1} and \eqref{eq:restricao-kernel}, and there are $M_j>0$ and $\mu_{j}>0$ such that
	 $$
	  |f_j(u)-f_j(v)|\leq\mu_j|u-v|\Es\text{and}\Es |f_{j}(u)|\leq M_j,\Es\forall u,v\in\er,\, j=1,\ldots,n.
	 $$
	
	If there exists $w=(w_1,\ldots,w_n)>0$ such that, for all $t\geq0$ and $i=1,\ldots,n$,	
	\begin{eqnarray}\label{4.3f}
		\un{a}_i\beta_iw_i>\ov{a}_i\sum_{j=1}^{n}\bigg[\dst|c_{ij}(t)|\mu_{j}\rho_j\lefteqn{w_j+|d_{ij11}(t)|\rho_{j}\mu_jw_j}\nonumber\\
		&
		+\dst\sum_{l=1}^n|d_{ijl2}(t)|\big(w_jM_j\mu_{j}\rho_j+w_lM_l\mu_{l}\rho_l\big)\bigg],
	\end{eqnarray}
	then the model \eqref{eq:L-O-Cohen-Grossberg-model} has an $\omega-$periodic solution which is globally exponentially stable.
\end{corolario}
\begin{proof}
    If we take in model \eqref{3.1c} $P=1$, $Q=2$, and, for each $i,j,l=1,\ldots,n$, $q=1,2$, the functions $F_i(u)=G_i(u)=u$, $\tau_{ij1}(t)=\til{\tau}_{ij1}(t)=0$, $h_{ijl1}(u_1,u_2)=f_j(\rho_ju_1)$, $c_{ij11}(t)=c_{ij}(t)$ $c_{ijl1}(t)=d_{ijl1}(t)=0$ for $l\neq1$,  $f_{ijl1}(u_1,u_2)=f_j(\rho_ju_1)$, $f_{ijl2}(u_1,u_2)=f_j(\rho_ju_1)f_j(\rho_ju_2)$, $g_{ijq}(u)=\til{g}_{ijq}(u)=u$, and $\til{\eta}_{ijq}(s)=\eta_{ijq}(s)$ defined by \eqref{eq:-defin-eta} for all $u_1,u_2\in\er$ and $s\leq0$, then we obtain model \eqref{eq:H-O-Cohen-Grossberg-model}.

    For all $u_1,u_2,v_1,,v_2\in\er$, we have
    $$
      |h_{ijl1}(u_1,u_2)-h_{ijl1}(v_1,v_2)|=|f_j(\rho_ju_1)-f_j(\rho_jv_1)|\leq\rho_j\mu_j|u_1-v_1|,
    $$
    $$
    |f_{ijl1}(u_1,u_2)-f_{ijl1}(v_1,v_2)|=|f_j(\rho_ju_1)-f_j(\rho_jv_1)|\leq\rho_j\mu_j|u_1-v_1|,
    $$
    and
    \begin{align*}
    |f_{ijl2}(u_1,u_2)-f_{ijl2}(v_1,v_2)|&=|f_j(\rho_ju_1)f_j(\rho_ju_2)-f_j(\rho_jv_1)f_j(\rho_ju_2)|\\
    &\leq M_j\rho_j\mu_j|u_1-v_1|+M_j\rho_j\mu_j|u_2-v_2|,
    \end{align*}
for all $i,j,l=1,\ldots,n$, thus hypothesis (H4**) holds. Condition \eqref{4.2} follows from \eqref{eq:restricao-kernel} and the inequality \eqref{4.3d} reads as \eqref{4.3f}. Finally, the result follows from Corollary \ref{cor:existencia+estabilidade-modelo-ampl-auto}.
\end{proof}

\begin{rem}
	In \cite{liu2011periodic}, sufficient conditions for the existence and global exponential stability of a $\omega-$periodic {solution} of \eqref{eq:H-O-Cohen-Grossberg-model} were presented. However, it is important to mention that the proof of the main result is not correct. Specifically, the way inequality \cite[(3.10)]{liu2011periodic} is obtained is problematic.
	In fact assuming the uniqueness of solution of \eqref{eq:H-O-Cohen-Grossberg-model} with initial condition $x_0=\psi$ for $\psi\in BC$, {denoting this solution} by $x(t,0,\psi)$, and defining $P:BC\to BC$ by $P(\psi)=x_\omega(\cdot,0,\psi)$, we always have
	\begin{align*}
	\|P^N(\psi_1)-P^N(\psi_2)\|&=\|x_{N\omega}(\cdot,0,\psi_1)-x_{N\omega}(\cdot,0,\psi_2)\|\\
	&=\sup_{s\leq0}\|x(N\omega+s,0,\psi_1)-x(N\omega+s,0,\psi_2)\|\\
	&\geq\|\psi_1-\psi_2\|
	\end{align*}
	for all $\psi_1, \psi_2\in BC$ and $N\in\en$, since model \eqref{eq:H-O-Cohen-Grossberg-model}  is $\omega-$periodic.
\end{rem}

\section{Numerical Example}

Here, we present a numerical example to illustrate the applicability of some new results given in this work.

The system
\begin{align}\label{eq:modelo-hopfield-exemplo}	
		x_1'(t)=&\left(\frac{1}{48}\sin\big(x_1(t)\big)+\frac{7}{48}\right)\bigg[-(9+\sin(t))x_1(t)+c\cos(t)\arctan\big(x_1(t-\sin(t))\big)\nonumber\\
		&\cdot \arctan\big(x_2(t-\cos(t))\big)+d\sin(t)\dst\int_0^{+\infty}\Ne^{-u}x_2(t-u)du+\cos(t)\bigg]\nonumber\\
		\\
		x_2'(t)=&\left(2+\cos\big(x_2(t)\big)\right)\bigg[-(2+\cos(t))x_2(t)+\hat{c}\sin(t)\arctan\big(x_1(t-\cos(t))\big)\nonumber\\
		&+\hat{d}\cos(t)\tanh\left(\dst\int_0^{+\infty}\Ne^{-u}x_1(t-u)du\right)\tanh\left(\dst\int_0^{+\infty}\Ne^{-u}x_2(t-u)du\right)+\Ne^{\sin(t)}\bigg]\nonumber
\end{align}
for $t\geq0$, where $c,d,\hat{c},\hat{d}\in\er$, is a $2\pi-$periodic example of a high-order Cohen-Grossberg neural network model.

Defining $\eta_{ij1},\til{\eta}_{ij1}:(-\infty,0]\to\er$ by
$$
  \eta_{ij1}(s)=\til{\eta}_{ij1}(s)=\int_{-\infty}^s\Ne^vdv,\Es s\in(-\infty,0],
$$
system \eqref{eq:modelo-hopfield-exemplo} is a particular situation of \eqref{3.1c}. However, \eqref{eq:modelo-hopfield-exemplo} is not a particular case of  \eqref{eq:H-O-Cohen-Grossberg-model}, thus the model studied in \cite{liu2011periodic} is not general enough to include \eqref{eq:modelo-hopfield-exemplo} as a particular example.

 Following the notations in \eqref{3.1c} and \eqref{eq:definicao-constantes}, we have $n=2$, $P=Q=1$, $\un{a}_1=\frac{1}{8}$, $\ov{a}_1=\frac{1}{6}$, $\un{a}_2=1$, $\ov{a}_2=3$, $\un{\beta}_1=8$, $\un{\beta}_2=1$, $\zeta_i=\varsigma_i=1$, $\gamma_{1121}^{(1)}=\gamma_{1121}^{(2)}=\frac{\pi}{2}$, $\gamma_{2111}^{(1)}=1$, $\mu_{1211}^{(2)}=\mu_{2121}^{(1)}=\mu_{2121}^{(2)}=1$,  $I_1(t)=\cos(t)$, $I_2(t)=\Ne^{\sin(t)}$, $\gamma_{1121}^{(1)}=\gamma_{1121}^{(2)}=\frac{\pi}{2}$, $\gamma_{2111}^{(1)}=1$, $\mu_{1211}^{(2)}=\mu_{2121}^{(1)}=\mu_{2121}^{(2)}=1$, $c_{1121}(t)=c\cos(t)$, $c_{2111}(t)=\hat{c}\sin(t)$,  $d_{1211}(t)=d\sin(t)$, $d_{2121}(t)=\hat{d}\cos(t)$, and all other $c_{ijl1}(t)=d_{ijl1}(t)=0$, for  $i,j,l=1,2$.

Consequently, example \eqref{eq:modelo-hopfield-exemplo} is $2\pi-$periodic and the matrix $\mathcal{M}$, defined in \eqref{def-M-matrix}, has the form
\begin{eqnarray*}
  \mathcal{M}=\left[\begin{array}{cc}
  	    1-\frac{\pi}{16}|c|&-\frac{1}{2}\left(\frac{\pi}{2}|c|+|d|\right)\\
  	    \\
  	   -3\left(\frac{\pi}{2}|\hat{c}|+|\hat{d}|\right)&1-3|\hat{d}|
  	    \end{array}\right].
\end{eqnarray*}

Condition \eqref{4.2} trivially holds with $\vartheta\in(0,1)$. Consequently, Corollary \ref{cor:crit-M-matrix} assures the existence and exponential stability of a $2\pi-$periodic solution of \eqref{eq:modelo-hopfield-exemplo}	in case $\mathcal{M}$ being a non-singular M-matrix. For example, if we consider $c=\frac{1}{\pi}$, $d=\frac{1}{100}$, $\hat{c}=\frac{1}{30\pi}$, $\hat{d}=\frac{1}{30}$, we have
\begin{eqnarray*}
	\mathcal{M}=\left[\begin{array}{cc}
		\frac{15}{16}&-\frac{101}{200}\\
		\\
		-\frac{1}{5}&\frac{9}{10}
	\end{array}\right],
\end{eqnarray*}
which is a non-singular M-matrix.

\noindent
{\bf Acknowledgments.}\\
This work was partially supported by Funda\c{c}\~ao para a Ci\^encia e a Tecnologia (Portugal) within the Projects UIDB/00013/2020, UIDP/00013/2020 of CMAT-UM (Jos\'e J. Oliveira), and Project UIDB/00212/2020 of CMA-UBI (Ahmed Elmwafy and C\'esar M. Silva).%, and Project UI/BD/151492/2021 (Ahmed Elmwafy).

%\renewcommand{\bibname}{References}
%\nocite{*}
%\bibliographystyle{myabbrvnat}
%\usepackage{natbib}   % omit 'round' option if you prefer square brackets
%\bibliographystyle{apalike}

%\bibliographystyle{unsrt}
%\bibliography{references}

\end{document}